\documentclass[12pt]{iopart}

\usepackage{iopams}
\expandafter\let\csname equation*\endcsname\relax
\expandafter\let\csname endequation*\endcsname\relax
\usepackage{amsmath,amssymb}

\usepackage[pdfencoding=auto]{hyperref}

\usepackage{color}
\usepackage{tikz}

\usepackage{braket}
\usepackage{bbm}
\usepackage{enumerate}

\usepackage{mathtools}

\usepackage{amsthm}
\theoremstyle{theorem}
\newtheorem{thm}{Theorem}
\newtheorem{prop}{Proposition}
\newtheorem{lemm}{Lemma}

\theoremstyle{definition}
\newtheorem{defi}{Definition}
\newtheorem{exam}{Example}
\theoremstyle{remark}
\newtheorem{rem}{Remark}

\newcommand{\defarrow}{\stackrel{\mathrm{def.}}{\Leftrightarrow}}

\newcommand{\natn}{\mathbb{N}}

\newcommand{\realn}{\mathbb{R}}
\newcommand{\cmplx}{\mathbb{C}}

\newcommand{\Nn}{\mathbb{N}_n}
\newcommand{\NN}{\mathbb{N}_N}
\newcommand{\Nm}{\mathbb{N}_m}
\newcommand{\Nd}{\mathbb{N}_d}

\newcommand{\cH}{\mathcal{H}}
\newcommand{\cK}{\mathcal{K}}

\newcommand{\cJ}{\mathcal{J}}

\newcommand{\oA}{\mathsf{A}}
\newcommand{\oB}{\mathsf{B}}

\newcommand{\oM}{\mathsf{M}}
\newcommand{\oN}{\mathsf{N}}

\newcommand{\oU}{\mathsf{U}}

\newcommand{\oAj}{\mathsf{A}^{(j)}}
\newcommand{\oAk}{\mathsf{A}^{(k)}}
\newcommand{\oBj}{\mathsf{B}^{(j)}}
\newcommand{\oBk}{\mathsf{B}^{(k)}}

\newcommand{\oMseq}{(\oM(j))_{j=0}^{m-1}}
\newcommand{\oNseq}{(\oN(k))_{k=0}^{n-1}}

\newcommand{\unit}{\mathbbm{1}}

\newcommand{\pp}{\preceq_{\mathrm{post}}}
\newcommand{\ppeq}{\sim_{\mathrm{post}}}
\newcommand{\CPpp}{\preceq_{\mathrm{CP}}}
\newcommand{\CPeq}{\sim_{\mathrm{CP}}}

\newcommand{\id}{\mathrm{id}}

\newcommand{\ordsim}{\mathord\sim}
\newcommand{\ordpreceq}{\mathord\preceq}
\newcommand{\ordpp}{\mathord\pp}

\newcommand{\ordCPpp}{\mathord\preceq_{\mathrm{CP}}}

\newcommand{\Mfin}{\mathfrak{M}_{\mathrm{fin}}}
\newcommand{\MfinH}{\mathfrak{M}_{\mathrm{fin}}(\mathcal{H})}

\newcommand{\evmfH}{\mathrm{EVM}_{\mathrm{fin}}(\DeH{})}

\newcommand{\cbit}{\mathrm{cbit}}
\newcommand{\Ocbit}{\Omega_{\mathrm{cbit}}}
\newcommand{\MfOc}{\Mfin (\Ocbit)}

\newcommand{\MfinO}{\mathfrak{M}_{\mathrm{fin}}(\Omega)}

\newcommand{\evm}{\mathrm{EVM}}
\newcommand{\evmfO}{\mathrm{EVM}_{\mathrm{fin}}(\Omega)}

\newcommand{\bCh}{\mathbf{Ch}}
\newcommand{\fC}{\mathfrak{C}}
\newcommand{\fCH}{\mathfrak{C} (\mathcal{H})}

\newcommand{\dimord}{\dim_{\mathrm{ord}}}
\newcommand{\dimordR}{\dim_{\mathrm{ord}, \, \realn}}
\newcommand{\dpc}{\mathrm{dpc}}

\newcommand{\BH}{\mathbf{B}(\mathcal{H})}
\newcommand{\BK}{\mathbf{B}(\mathcal{K})}
\newcommand{\BKF}{\mathbf{B}(P_F\mathcal{K})}
\newcommand{\BJ}{\mathbf{B}(\mathcal{J})}

\newcommand{\BCn}{\mathbf{B}(\mathbb{C}^n)}

\newcommand{\TcH}{\mathbf{T}(\mathcal{H})}
\newcommand{\TcK}{\mathbf{T}(\mathcal{K})}

\newcommand{\BsaH}{\mathbf{B}_{\mathrm{sa}}(\mathcal{H})}

\newcommand{\TsaH}{\mathbf{T}_{\mathrm{sa}}(\mathcal{H})}

\newcommand{\BpH}{\mathbf{B}_+(\mathcal{H})}

\newcommand{\TpH}{\mathbf{T}_+(\mathcal{H})}

\newcommand{\DeH}{\mathbf{D}(\mathcal{H})}
\newcommand{\DeK}{\mathbf{D}(\mathcal{K})}

\newcommand{\EO}{\mathcal{E}(\Omega)}

\newcommand{\Pg}{P_\mathrm{g}}
\newcommand{\E}{\mathcal{E}}

\newcommand{\F}{\mathcal{F}}

\newcommand{\FI}{\mathbb{F}(I)}

\newcommand{\uwarrow}{\xrightarrow{\mathrm{uw}}}

\DeclareMathOperator*{\uwlim}{uw-lim}

\begin{document}

\title[Infinite dimensionality of the post-processing order of measurements]
{Infinite dimensionality of the post-processing order of measurements on a general state space}

\author{Yui Kuramochi}

\address{Department of Physics, Faculty of Science, Kyushu University, 744 Motooka, Nishi-ku, Fukuoka, Japan}
\ead{yui.tasuke.kuramochi@gmail.com}
\vspace{10pt}
\begin{indented}
\item[]\today
\end{indented}

\begin{abstract}
For a partially ordered set $(S, \mathord\preceq)$, the order (monotone) dimension is the minimum cardinality of total orders (respectively, real-valued order monotone functions) on $S$ that characterize the order $\preceq$.
In this paper we consider an arbitrary generalized probabilistic theory and the set of finite-outcome measurements on it, which can be described by effect-valued measures, equipped with the classical post-processing orders.
We prove that the order and order monotone dimensions of the post-processing order are (countably) infinite if the state space is not a singleton (and is separable in the norm topology).
This result gives a negative answer to the open question for quantum measurements posed in [Guff T \textit{et al.\/} 2021 \textit{J.\ Phys.\ A: Math.\ Theor.} \textbf{54} 225301].
We also consider the quantum post-processing relation of channels with a fixed input quantum system described by a separable Hilbert space $\mathcal{H}$ and show that the order (monotone) dimension is countably infinite when $\dim \mathcal{H} \geq 2$.
\end{abstract}

\vspace{2pc}

\section{Introduction} \label{sec:intro}

Order structures are ubiquitous in broad range of physical theories.
Examples of them are
the adiabatic accessibility relation in thermodynamics~\cite{giles1964mathematical,LIEB19991}, 
the convertibility relation by local operations and classical communication (LOCC) of bipartite quantum states~\cite{RevModPhys.81.865,nielsen_chuang_2010}, the post-processing relation of quantum measurements, or positive operator-valued measures (POVMs)~\cite{Martens1990,Dorofeev1997349,buscemi2005clean,jencova2008,10.1063/1.4934235,10.1063/1.4961516,Guff_2021}, and so on.
These relations are related to (quantum) resource theories~\cite{RevModPhys.91.025001}, which are recently being intensively studied in the area of quantum information.

In thermodynamics, the adiabatic accessibility relation is characterized by a single function, the entropy.
This means that a thermodynamic state is adiabatically convertible to another one if and only if the entropy increases.
From this, we can see that the adiabatic accessibility relation is a total order.
On the other hand, the LOCC convertibility relation of finite-dimensional bipartite pure state is not a total order, while the well-known characterization of this relation by the majorization order (\cite{nielsen_chuang_2010}, Theorem~12.15) suggests that we have still finite number of order monotones (i.e.\ order-preserving real-valued functions) that characterize the order.

How about the case of the post-processing relation of POVMs?
It is shown by many authors~\cite{buscemi2016degradable,PhysRevLett.122.140403,Guff_2021,kuramochi2020compact} that the state discrimination probabilities characterize the post-processing order (Proposition~\ref{prop:bss}), while they are uncountably infinite as pointed out in \cite{Guff_2021}.
These in mind, the authors of \cite{Guff_2021} asked if we can characterize the post-processing order of POVMs by a \textit{finite} number of order monotones, as in the case of the LOCC convertibility relation of bipartite pure states.

In this paper we give a negative answer to this open question.
More strongly, we show that any choice of finite number of order monotones or total orders cannot characterize the post-processing relation of measurements on any non-trivial generalized probabilistic theory (GPT)~\cite{1751-8121-47-32-323001,ddd.uab.cat:187745,kuramochi2020compact}, which is a general framework of physical theories and contains the quantum and classical theories as special examples.
Moreover, we demonstrate that a \textit{countable} number of order monotones characterizing the post-processing order exists if the state space is separable in the norm topology. 
As a corollary, we also prove a similar statement for the quantum post-processing relation for quantum channels with a fixed input Hilbert space.

The main theorems (Theorems~\ref{thm:main1} and \ref{thm:main2}) and their proofs are based on the notion of the \textit{order dimension} (Definition~\ref{def:dim}) known in the area of order theory~\cite{Dushnik1941,Hiraguti1955,BA18230795}.
The order dimension of an order is defined as the minimum number of total orders that characterize the order and roughly quantifies the complexity of the order or its deviation from a simple total order.

In the main part, we also introduce another related quantity called the \textit{order monotone dimension} of an order (Definition~\ref{def:dim}).
This is defined as the minimum number of order monotones that characterize the order and directly connected to the open question in \cite{Guff_2021}.
The order monotone dimension is shown to be always greater than or equal to the order dimension (see Lemma~\ref{lemm:leq}).
The present work is the first attempt to evaluate these dimensions of orders appearing in quantum information.

The rest of this paper is organized as follows.
In Section~\ref{sec:prel} we introduce some preliminaries and definitions.
In Section~\ref{sec:main} we state and prove the main theorems (Theorems~\ref{thm:main1} and \ref{thm:main2}).
We then conclude the paper in Section~\ref{sec:conclusion}.

\section{Preliminaries} \label{sec:prel}
In this section, we give preliminaries on the order theory and post-processing relations of measurements and quantum channels and fix the notation.

In this paper, we denote by $\natn = \set{1,2, \dots}$, $\realn$, and $\cmplx$ the sets of natural, real, and complex numbers, respectively. 
For each $m \in \natn $, we write as $\Nm := \{ 0,1,\dots , m-1 \}$.
The vector spaces in this paper are over $\realn$ unless otherwise stated.

\subsection{Order theory} \label{subsec:ord}
In this paper we identify each binary relation $R$ on a set $S$ with its graph, that is, $R$ is the subset of $S \times S$ consisting of pairs satisfying the relation $R .$
For a binary relation $R \subseteq S \times S ,$ the relation $(x,y) \in R $ $(x,y \in S)$ is occasionally written as $x R y ,$ which is consistent with the common notation.
If $A$ is a subset of $S$ and $R$ is a binary relation on $S ,$ we define the \textit{restriction} of $R$ to $A$ by $ R\rvert_A  := R \cap (A \times A) . $

Let $R$ be a binary relation on a set $S .$
We consider the following conditions for $R.$
\begin{itemize}
\item
$R$ is \textit{reflexive} $:\defarrow$ $x R x$ $(\forall x \in S).$
\item
$R$ is \textit{symmetric} $:\defarrow$ $x R y$ implies $yRx$ $(\forall x ,y \in S).$
\item
$R$ is \textit{antisymmetric} $:\defarrow$ $x R y$ and $yRx$ imply $x=y$ $(\forall x ,y \in S).$
\item
$R$ is \textit{transitive} $: \defarrow$ $xRy$ and $yRz$ imply $xRz$ $(\forall x ,y, z \in S).$
\item
$R$ is \textit{total} $:\defarrow$ either $xRy$ or $yRx$ holds $(\forall x, y \in S) .$
\end{itemize}
A reflexive, transitive binary relation is called a \textit{preorder}.
An antisymmetric preorder is called a \textit{partial order}.
A total partial order is just called a \textit{total order} (or a \textit{linear order}).
If $R$ is a partial order on a set $S ,$ $(S,R)$ is called a \textit{partially ordered set}, or a \textit{poset}.
A poset $(S, R)$ is called a \textit{chain} if $R$ is a total order.

If $P$ is a preorder on $S ,$ then the relation $\sim$ defined via
\[
	x \sim y : \defarrow  \text{$xPy$ and $yPx$} \quad (x,y \in S)
\]
is an equivalence relation (i.e.\ a reflexive, symmetric, and transitive relation).
If we denote by $[x]$ the equivalence class to which $x \in S$ belongs, then we may define a binary relation $R$ on the quotient space $S/ \ordsim$ via
\[
	[x] R [y] :\defarrow x Py
\]
and $R$ is a partial order on $S/ \ordsim$.

A binary relation $P$ on $S$ is said to be an \textit{extension} of a binary relation $R$ on $S$ if $R \subseteq P ,$ i.e.\ $xRy$ implies $xPy$ for every $x,y \in S .$
If $L$ is an extension of a binary relation $R $ and is a total order, then $L$ is called a \textit{linear extension} of $R$.

A family $\mathcal{L}$ of binary relations on a set $S$ is said to \textit{realize} a binary relation $R$ on $S$ (or \textit{realize} $(S,R)$) if $R = \bigcap \mathcal{L} ,$ 
i.e.\ for every $x,y \in S $
\[
	x R y \iff [ x L y \quad (\forall L \in \mathcal{L})].
\]

Let $(S , \ordpreceq)$ be a poset.
A real-valued function $f \colon S \to \realn$ is called an \textit{order monotone}, or a \textit{resource monotone}, (with respect to $\ordpreceq$) if for every $x,y \in S$
\begin{equation}
	x \preceq y \implies f(x) \leq f(y) .
	\notag
\end{equation}
A family $\F$ of real-valued functions on $S$ is said to \textit{characterize}
(or to be a \textit{complete family} of) $(S, \ordpreceq) $ if for every $x,y \in S$
\begin{equation}
	x \preceq y \iff [f(x) \leq f(y) \quad (\forall f \in \F )] .
	\label{eq:condiF}
\end{equation}
If \eqref{eq:condiF} holds, then each $f\in \F$ is necessarily an order monotone.

The following notions of the order and order monotone dimensions play the central role in this paper.

\begin{defi} \label{def:dim}
Let $(S , \ordpreceq)$ be a poset.
\begin{enumerate}
\item 
By the \textit{order dimension} \cite{Dushnik1941,Hiraguti1955} of $(S , \ordpreceq) ,$ written as $\dimord (S , \ordpreceq ) ,$ we mean the minimum cardinality~\cite{halmos1960naive} $|\mathcal{L} |$ of a family $\mathcal{L}$ of linear extensions of $\ordpreceq$ that realizes $(S ,\ordpreceq ).$
Here $|A|$ denotes the cardinality of a set $A .$
\item
By the \textit{order monotone dimension} of $(S , \ordpreceq) ,$ written as $\dimordR (S , \ordpreceq ) ,$ we mean the minimum cardinality $|\F |$ of a family $\F$ of order monotones on $S$ that characterizes $(S, \ordpreceq ) .$
\end{enumerate}
\end{defi}
The order monotone dimension of a poset $(S , \ordpreceq)$ is well-defined.
Indeed, according to \cite{COECKE201659} (Proposition~5.2), if we define the order monotone
\begin{equation}
	M_a (x) := 
	\begin{cases}
	1 & (\text{when $a \preceq x$}); \\
	0 &(\text{otherwise})
	\end{cases}
	\quad
	(x \in S) 
	\notag
\end{equation}
for each $a\in S ,$
then the family $\{ M_a \}_{a \in S}$ characterizes $(S ,\preceq ) .$
Hence $\dimordR (S , \ordpreceq)$ is well-defined and at most $|S| .$
The well-definedness of the order dimension is proved in \cite{Dushnik1941} (Theorem~2.32) by using the Szpilrajn extension theorem \cite{SzpilrajnSurLD}.
We also note that, since the cardinals are well-ordered~\cite{halmos1960naive}, we can always take a family $\mathcal{F}$ of monotones that characterizes $\preceq$ and the cardinality $|\mathcal{F}|$ is minimum, i.e.\ $|\F | = \dimordR (S , \ordpreceq) .$
A similar statement for the order dimension is also true.
We will prove in Lemma~\ref{lemm:leq} that the order monotone dimension is always greater than or equal to the order dimension.

\subsection{General probabilistic theory} \label{subsec:gpt}
A general probabilistic theory (GPT) with the no-restriction hypothesis~\cite{PhysRevA.87.052131} is mathematically described by the following notion of the base-norm Banach space.

\begin{defi} \label{def:gpt}
A triple $(V, V_+ , \Omega)$ is called a \textit{base-norm Banach space} if the following conditions hold.
\begin{enumerate}
\item
$V$ is a real vector space.
\item
$V_+$ is a positive cone of $V$, i.e.\ $\lambda V_+ \subseteq V_+$ $(\forall \lambda \in [0,\infty))$, $V_+ + V_+ \subseteq V_+$, and $V_+ \cap (-V_+) = \{ 0\}$ hold.
We define the linear order on $V$ induced from $V_+$ by
\begin{equation*}
	x \leq y : \defarrow y-x \in V_+ \quad (x,y \in V) .
\end{equation*}
\item
$V_+$ is generating, i.e.\ $V = V_+ + (-V_+)$.
\item
$\Omega$ is a base of $V_+$, i.e.\ $\Omega$ is a convex subset of $V_+$ and for every $x \in V_+$ there exists a unique $\lambda \in [0,\infty )$ such that $x \in \lambda \Omega $.
\item
We define the base-norm on $V$ by
\begin{equation*}
	\| x \| := \inf \set{\alpha + \beta | x = \alpha \omega_1 + \beta \omega_2 ; \, \alpha , \beta \in [0,\infty ); \, \omega_1 , \omega_2 \in \Omega}  
	\quad (x \in V).
\end{equation*}
We require that the base-norm $\| \cdot \|$ is a complete norm on $V$.
\end{enumerate}
If these conditions are satisfied, $\Omega$ is called a \textit{state space} and each element $\omega \in \Omega$ is called a \textit{state}.
\end{defi}

The reader can find in \cite{ddd.uab.cat:187745} (Chapter~1) how the notion of the base-norm Banach space is derived from operationally natural requirements on the GPT.

Let $(V , V_+ , \Omega)$ be a base-norm Banach space.
We denote by $V^\ast$ the continuous dual of $V$ (i.e.\ the set of norm-continuous real linear functionals on $V$) equipped with the dual norm
\begin{equation*}
	\| f\| := \sup_{x \in V  , \, \| x\| \leq 1} |f(x)| \quad (f\in V^\ast) .
\end{equation*}
For each $f \in V^\ast$ and each $x \in V$ we occasionally write as $\braket{f,x} := f(x)$.
The dual positive cone $V^\ast_+$ of $V^\ast$ is defined by
\begin{equation*}
	V^\ast_+ := \set{f \in V^\ast | \braket{f,x} \geq 0 \, (\forall x \in V_+)}
\end{equation*}
and the dual linear order by
\begin{equation*}
	f \leq g : \defarrow g - f \in V^\ast_+ \Leftrightarrow [f(x) \leq g(x) \quad (\forall x \in V_+)] \qquad (f,g \in V^\ast).
\end{equation*}
It can be shown that there exists a unique positive element, called the unit element, $u_\Omega \in V^\ast_+$ such that $\braket{u_\Omega,\Omega} = 1.$
Then the dual norm on $V^\ast$ coincides with the order unit norm of $u_\Omega$:
\begin{equation*}
	\| f \| = 
	\inf \set{\lambda \in [0,\infty) | - \lambda u \leq f \leq \lambda u}
	\quad (f \in V^\ast).
\end{equation*}
An element $e \in V^\ast$ satisfying $0 \leq e \leq u_\Omega$ is called an \textit{effect} (on $\Omega$).
The set of effects on $\Omega$ is denoted by $\EO$.

In the main part of the paper, we will consider the following examples of quantum and classical theories.

\begin{exam}[Quantum theory] \label{ex:quantum}
Let $\cH$ be a complex Hilbert space.
We write the inner product of $\cH$ as $\braket{ \cdot | \cdot }$ which is antilinear and linear in the first and second components, respectively, and the complete norm as $\| \psi \| := \braket{\psi | \psi}^{1/2} .$
The sets of bounded and trace-class linear operators on $\cH$ are denoted by $\BH$ and $\TcH ,$ respectively.
The self-adjoint and positive parts of these sets are defined by
\begin{gather*}
	\BsaH := \set{a \in \BH | a = a^\ast} , \\
	\BpH := \set{a \in \BH | \braket{\psi | a \psi} \geq 0 \, (\forall \psi \in \cH)}, \\
	\TsaH := \set{a \in \TcH | a = a^\ast}, \\
	\TpH := \set{a \in \TcH | \braket{\psi | a \psi} \geq 0 \, (\forall \psi \in \cH)} ,
\end{gather*}
where $a^\ast$ denotes the adjoint operator of $a \in \BH$.
The uniform and the trace norms are respectively defined by
\begin{gather*}
	\| a \| := \sup_{\psi \in \cH , \, \| \psi \| \leq 1} \| a \psi \|  
	\quad (a \in \BH) ,
	\\
	\| b \|_1 : = \tr (\sqrt{b^\ast b}) \quad (b \in \TcH),
\end{gather*}
where $\tr(\cdot)$ denotes the trace.
A non-negative trace-class operator $\rho$ satisfying the normalization condition $\tr (\rho) =1$ is called a density operator.
The set of density operators on $\cH$ is denoted by $\DeH$.

In the GPT framework in Definition~\ref{def:gpt}, the quantum theory corresponds to the case
$(V,V_+ , \Omega) = (\TsaH , \TpH , \DeH)$.
The base-norm on $\TsaH$ then coincides with the trace norm.
The continuous dual space $\TsaH^\ast$, the dual norm on it, the dual positive cone $\TsaH^\ast_+$, and the unit element $u_{\DeH}$ are respectively identified with $\BsaH$, the uniform norm, $\BpH$, and the identity operator $\unit_{\cH}$ on $\cH$ by the duality
\begin{equation*}
	\braket	{a , b} = \tr (ab)  \quad (a \in \BsaH ; \, b \in \TsaH).
\end{equation*}
By this duality, we identify $\TsaH^\ast$ with $\BsaH$.
\end{exam}

\begin{exam}[Discrete classical theory] \label{ex:classical}
Let $X$ be a non-empty set.
We define $\ell^1(X)$ and $\ell^\infty(X)$ and their positive parts by 
\begin{gather*}
	\ell^1(X) := \set{f = (f(x))_{x\in X} \in \realn^X | \| f\|_1 < \infty}, \\
	\ell^\infty(X) := \set{f = (f(x))_{x\in X} \in \realn^X | \| f\|_\infty < \infty}, \\
	\ell^1_+(X) := \set{(f(x))_{x\in X} \in \ell^1(X) | f(x) \geq 0 \, (\forall x \in X)}, \\
	\ell^\infty_+(X) := \set{(f(x))_{x\in X} \in \ell^\infty (X) | f(x) \geq 0 \, (\forall x \in X)},
\end{gather*}
where for $f \in \realn^X$
\begin{gather*}
	\|f \|_1 := \sum_{x\in X} | f(x) |, \\
	\| f \|_\infty := \sup_{x\in X} |f(x)| 
\end{gather*}
are respectively the $\ell^1$- and the $\ell^\infty$-norms.
We also define the simplex of the probability distributions on $X$ by
\begin{equation*}
	\mathcal{P} (X) := \Set{(p (x))_{x\in X} \in \ell^1_+(X) | \sum_{x\in X} p(x) =1 } .
\end{equation*}

In the framework of the GPT in Definition~\ref{def:gpt}, a discrete classical theory corresponds to the case $(V,V_+, \Omega) = (\ell^1 (X) , \ell^1_+(X), \mathcal{P}(X))$.
The base-norm on $\ell^1(X)$ coincides with the $\ell^1$-norm $\| \cdot \|_1$.
The continuous dual $\ell^1(X)^\ast$, the dual norm on $\ell^1(X)^\ast$, the dual positive cone $\ell^1(X)^\ast_+$, and the unit element $u_{\mathcal{P}(X)}$ are respectively identified with $\ell^\infty(X)$, the $\ell^\infty$-norm $\| \cdot \|_\infty$, $\ell^\infty_+(X)$, and the constant function
$
	1_X := (1)_{x\in X}
$
by the duality
\begin{equation*}
	\braket{f,g} = \sum_{x\in X} f(x) g(x) \quad (f\in \ell^\infty (X) ; \, g \in \ell^1(X)) .
\end{equation*}
\end{exam}

\subsection{Post-processing relation of measurements on a GPT} \label{subsec:gptpost}
Now we fix a base-norm Banach space $(V, V_+ , \Omega)$ corresponding to a GPT.

For a natural number $m \in \natn$, a finite sequence $(\oM(k))_{k =0}^{m-1} \in (V^\ast)^m$ is called an \textit{($m$-outcome) effect-valued measure} (EVM) (on $\Omega$) if $\oM(k) \geq 0$ $(k\in \Nm)$ and $\sum_{k =0}^{m-1}\oM(k) = u_\Omega $ hold.
We denote by $\evm_m (\Omega) $ the set of $m$-outcome EVMs.
We also define $\evmfO := \bigcup_{m \in \natn} \evm_m (\Omega)$, which is the set of finite-outcome EVMs on $\Omega$.
For each EVM $\oM = (\oM(k))_{k =0}^{m-1}$ and each state $\omega \in \Omega$, the sequence $(p^\oM_\omega (k))_{k =0}^{m-1}$ defined by 
\begin{equation*}
	p^\oM_\omega (k) := \braket{\oM(k) , \omega} \quad (k \in \Nm)
\end{equation*}
is a probability distribution.
In the physical context, $p^\oM_\omega$ is the outcome probability distribution of the measurement $\oM$ when the state of the system is prepared to be $\omega$.

Let $\oM = \oMseq$ and $\oN = \oNseq$ be EVMs on $\Omega$.
$\oM$ is said to be a \textit{post-processing} of (or, \textit{less or equally informative} than) $\oN$ \cite{Martens1990,Dorofeev1997349,buscemi2005clean,jencova2008,10.1063/1.4934235,10.1063/1.4961516}, written as $\oM \pp \oN $, if there exists a matrix $(p(j|k))_{j \in \Nm , \, k \in \Nn}$ such that 
\begin{gather}
	p(j|k) \geq 0  \quad (j \in \Nm , \, k \in \Nn) , 
	\label{eq:MK1} 
	\\
	\sum_{j=0}^{m-1} p(j|k) =1 \quad ( k \in \Nn) ,
	\label{eq:MK2}
	\\
	\oM(j) = \sum_{k=0}^{n-1} p(j|k) \oN(k) \quad (j \in \Nm) .
	\notag
\end{gather}
A matrix $(p(j|k))_{j \in \Nm , \, k \in \Nn}$ satisfying \eqref{eq:MK1} and \eqref{eq:MK2} is called a \textit{Markov matrix}.
The relation $\oM \pp \oN $ means that the measurement $\oM$ is realized if we first perform $\oN ,$ which gives a measurement outcome $k \in \Nn$, then randomly generate $j \in \Nm$ according to the probability distribution $(p(j|k))_{ j \in \Nm }$, forget the original measurement outcome $k$, and finally record $j$ as the measurement outcome.
We also say that $\oM$ and $\oN$ are \textit{post-processing equivalent} (or \textit{equally informative}), written as $\oM \ppeq \oN $, if both $\oM \pp \oN$ and $\oN \pp \oM$ hold.
The binary relations $\pp$ and $\ppeq$ are respectively preorder and equivalence relations on $\evmfO .$

We write as $\MfinO := \evmfO / \mathord\ppeq $  
and, for each $\oM \in \evmfO ,$ denote by $[\oM]$ the equivalence class to which $\oM$ belongs.
We define the binary relation $\pp$ on $\MfinO$ by
\begin{equation*}
	[\oM] \pp [\oN] :\defarrow \oM \pp \oN \quad ([\oM] , [\oN] \in \MfinO) .
\end{equation*}
Then $(\MfinO , \ordpp )$ is a poset.

An EVM $\oM = (\oM(k))_{k =0}^{m-1}$ is called \textit{trivial} if each component $\oM(k)$ is proportional to the unit $u_\Omega .$
The equivalence class $[\oM]$ of a trivial EVM $\oM$ is the minimum element of the poset $(\MfinO , \pp)$, i.e.\ $[\oM] \pp [\oN]$ for all $[\oN] \in \MfinO$.

The post-processing relation $\pp$ on $\MfinO$ is characterized by the state discrimination probabilities defined as follows.
A finite sequence $\E = (\rho_k)_{k=0}^{N-1} \in V_+^N$ of non-negative elements in $V$ is called an \textit{ensemble} (on $\Omega$) if the normalization condition $\sum_{k=0}^{N-1} \braket{u_\Omega , \rho_k}=1$ holds.
For an ensemble $\E = (\rho_k)_{k=0}^{N-1}$ and an EVM $\oM = \oMseq$ on $\Omega$, 
we define the \textit{state discrimination probability} by
\begin{align}
	\Pg (\E ; \oM) 
	&:= \sup_{\text{$(p(k|j))_{k \in \NN, \, j \in \Nm}$: Markov matrix}}
	\sum_{k \in \NN , \, j \in \Nm} p(k|j) \braket{ \oM (j) ,\rho_k} 
	\label{eq:Pgdef}
	\\
	&= \sup_{\oA \in \evm (N; \oM)} \sum_{k \in \NN} \braket{\oA(k), \rho_k}, \notag
\end{align}
where
\begin{equation}
	\evm (N ; \oM)
	:= \set{\oA = \left( \oA (k) ) \right)_{k=0}^{N-1} \in \evm_N (\Omega) | \oA \pp \oM }  ,
	\notag
\end{equation}
is the set of $N$-outcome EVMs obtained by post-processing $\oM .$
The ensemble $\E = (\rho_k)_{k=0}^{N-1}$ corresponds to the situation in which the system is prepared in the state $\braket{u_\Omega , \rho_k}^{-1}\rho_k$ with the probability $\braket{u_\Omega , \rho_k}.$
The state discrimination probability \eqref{eq:Pgdef} is the optimum average probability of the correct guessing of the index $k$ of the ensemble when we are given the measurement outcome $j$ of $\oM .$
The maximum of the optimization problem in the RHS of \eqref{eq:Pgdef} is attained when $p(k|j) = \delta_{k, k(j)} ,$ where $k(j)$ is chosen so that $k(j) \in \arg \max_{k \in \NN} \braket{\oM(j), \rho_k}$, i.e.\ when the maximum likelihood estimation is adopted. 
The optimal value is then given by
\begin{equation}
	\Pg (\E ; \oM) = \sum_{j \in \Nm} \max_{k \in \NN} \braket{\oM(j) , \rho_k} 
	\label{eq:mle}
\end{equation}
(\cite{kuramochi2020compact}, Lemma~5).

The following Blackwell-Sherman-Stein (BSS) theorem for EVMs~\cite{buscemi2016degradable,PhysRevLett.122.140403,Guff_2021,kuramochi2020compact} can be regarded as the generalization of the corresponding BSS theorem for statistical experiments known in mathematical statistics~\cite{lecam1986asymptotic,torgersen1991comparison}.
\begin{prop}[BSS theorem for EVMs] \label{prop:bss}
Let $(V, V_+ , \Omega)$ be a base-norm Banach space and let $\oM$ and $\oN$ be finite-outcome EVMs on $\Omega$.
Then $\oM \pp \oN$ if and only if $\Pg (\E ; \oM) \leq \Pg (\E; \oN)$ for all ensemble $\E$ on $\Omega$.
\end{prop}
The BSS theorem for continuous EVMs on a possibly infinite-dimensional GPT is proved in \cite{kuramochi2020compact}.
We give a straightforward proof of Proposition~\ref{prop:bss} in \ref{app:bssevm}.
Proposition~\ref{prop:bss} implies that the functions
\begin{equation}
	\MfinH  \ni [\oM]  \mapsto \Pg (\E ; \oM) =: \Pg (\E ; [\oM]) \in \realn
	\quad 
	(\text{$\E$: ensemble})
	\notag 
\end{equation}
are well-defined and the family 
\[
\set{\Pg (\E ; \cdot) \mid \text{$\E$ is an ensemble on $\Omega$}}
\] 
characterizes the poset $(\MfinO , \ordpp) .$

\subsection{Post-processing relation of quantum channels} \label{subsec:cp}
Let $\cH$ and $\cK$ be complex Hilbert spaces and let $\Gamma \colon \BK \to \BH$ be a complex linear map.
We have the following definitions.
\begin{itemize}
\item
$\Gamma$ is \textit{unital} $:\defarrow$ $\Gamma (\unit_{\cK}) = \unit_{\cH}$.
\item
$\Gamma$ is \textit{positive} $:\defarrow$ $a \geq 0$ implies $\Gamma (a) \geq 0$ for every $a \in \BK$.
\item
$\Gamma$ is \textit{completely positive} (\textit{CP}) \cite{1955stinespring,paulsen_2003,takesakivol1} $: \defarrow$ the product linear map $\Gamma \otimes \id_n \colon \mathbf{B}(\cK \otimes \cmplx^n) \to \mathbf{B}(\cH \otimes \cmplx^n)$ is positive for all $n \in \natn ,$ where $\id_n \colon \BCn \to \BCn$ is the identity map on $\BCn .$
\item
$\Gamma$ is a \textit{channel} (in the Heisenberg picture) $:\defarrow$ $\Gamma$ is unital and CP.
\item
For positive $\Gamma ,$ $\Gamma$ is normal $:\defarrow$ $\sup_i \Gamma (a_i) = \Gamma (\sup_i a_i)$ for every upper bounded increasing net $(a_i)$ in $\BK .$
This condition is equivalent to the ultraweak continuity of $\Gamma $, which means that $b_j \uwarrow b$ implies $\Gamma (b_j) \uwarrow \Gamma(b)$ for every net $(b_j)$ and every element $b$ in $\BK $, where $\uwarrow$ denotes the ultraweak (or $\sigma$-weak~\cite{takesakivol1}) convergence.
\end{itemize}
For a channel $\Gamma \colon \BK \to \BH ,$ the Hilbert spaces $\cH$ and $\cK$ are called respectively the input and output Hilbert spaces of $\Gamma$. 

If $\Gamma \colon \BK \to \BH$ is positive and normal, then there exists the unique positive linear map $\Gamma_\ast \colon \TcH \to \TcK ,$ called the predual map of $\Gamma$, such that
\begin{equation}
	\tr (\rho \Gamma (a)) = \tr (\Gamma_\ast (\rho) a) 
	\quad (\rho \in \TcH ; a \in \BK) .
	\label{eq:predual}
\end{equation}
Conversely if $\Gamma_\ast \colon \TcH \to \TcK$ is a positive linear map then there exists unique normal positive linear map $\Gamma$ satisfying \eqref{eq:predual}.
The predual of a normal channel describes the state change of the system (channel in the Schr\"odinger picture).

Let $\Gamma \colon \BK \to \BH$ and $\Lambda \colon \BJ \to \BH$ be normal channels with the same input Hilbert space $\cH .$
We define the post-processing relations for channels as follows.
\begin{itemize}
\item
$\Gamma \CPpp \Lambda$ ($\Gamma$ is less or equally informative than $\Lambda$) $:\defarrow$ there exists a normal channel $\Psi \colon \BK \to \BJ$ such that $\Gamma = \Lambda \circ \Psi .$
It is known that $\Gamma \CPpp \Lambda$ holds if and only if there exists a (not necessarily normal) channel $\Phi \colon \BK \to \BJ$ such that $\Gamma = \Lambda \circ \Phi $
(\cite{gutajencova2007}, Lemma~3.12; \cite{kuramochi2018incomp}, Theorem~2).
\item
$\Gamma \CPeq \Lambda$ ($\Gamma$ and $\Lambda$ are equally informative) 
$:\defarrow$ $\Gamma \CPpp \Lambda$ and $\Lambda \CPpp \Gamma .$
\end{itemize}
The binary relations $\mathord\CPpp$ and $\mathord\CPeq$ are respectively a preorder and an equivalence relation on the class $\bCh (\to \BH)$ of normal channels with a fixed input Hilbert space $\cH$. 
It can be shown that there exist a set $\fCH$ and a class-to-set surjection
\begin{equation}
	\bCh (\to \BH) \ni \Gamma \mapsto [\Gamma] \in \fCH
	\label{eq:Chmap}
\end{equation}
such that 
\[ \Gamma \CPeq \Lambda \iff [\Gamma] = [ \Lambda ] \] 
for every $\Gamma , \Lambda \in \bCh (\to \BH) . $
This follows from a more general result for normal channels with arbitrary input and output von Neumann algebras (\cite{kuramochi2020directed}, Section~3.3).
We fix such a set $\fCH$ and a map \eqref{eq:Chmap}.
We also define the post-processing relation on $\fCH$ by 
\begin{equation*}
	[\Gamma] \CPpp [\Lambda] :\defarrow \Gamma \CPpp \Lambda 
	\quad ([\Gamma ] , [\Lambda ] \in \fCH ) .
\end{equation*}
Then $(\fCH , \ordCPpp)$ is a poset.


Let $\cH$ be a complex Hilbert space, let $\E = (\rho_k)_{k=0}^{n-1}$ be an ensemble on $\DeH$, and let $\Gamma \colon \BK \to \BH $ a channel (or more generally, a unital positive map).
We define the state discrimination probability by 
\begin{equation}
	\Pg (\E ; \Gamma)
	:= \sup_{(\oM(k))_{k=0}^{n-1} \in \evm_n (\DeK)} 
	\sum_{k=0}^{n-1} \braket{\Gamma (\oM(k)), \rho_k} .
	\label{eq:Pgch}
\end{equation}
This quantity is the maximal state discrimination probability of the index $k$ of the ensemble $\E$ when the operation $\Gamma$ is performed on the system whose state is prepared according to $\E$, and then an optimal measurement $\oM$ on the output space $\cK $ is performed. 

The post-processing relation for channels is characterized by the state discrimination probabilities \textit{with quantum side information} as shown in the following BSS-type theorem.

\begin{prop}[BSS theorem for channels] \label{prop:qbss}
Let $\Gamma \colon \BK \to \BH$ and $\Lambda \colon \BJ \to \BH$ be normal channels.
Then the following conditions are equivalent.
\begin{enumerate}[(i)]
\item \label{it:qbss1}
$\Gamma \CPpp \Lambda .$
\item \label{it:qbss2}
For every $n \in \natn$ and every ensemble $\E$ on $\mathbf{D} (\cH \otimes \cmplx^n)$
\begin{equation}
	\Pg ( \E ; \Gamma \otimes \id_n) \leq \Pg (\E  ;\Lambda \otimes \id_n)  
	\label{eq:Pgleq}
\end{equation}
holds.
\end{enumerate}
\end{prop}
Proposition~\ref{prop:qbss} for finite-dimensional channels is proved in \cite{chefles2009quantum} by using Shmaya\rq{}s theorem \cite{Shmaya_2005}.
In \ref{app:bss}, we give a proof of Proposition~\ref{prop:qbss} based on another (infinite-dimensional) BSS theorem for normal positive maps obtained in \cite{luczak2019}.

Proposition~\ref{prop:qbss} implies that the function
\begin{equation*}
	\fCH \ni [\Gamma] \mapsto \Pg (\E ; \Gamma \otimes \id_n)
	=: \Pg^{(n)} (\E ; [\Gamma]) \in \realn
\end{equation*}
is a well-defined order monotone for every $n\in \natn $ and every ensemble $\E$ on $\mathbf{D}(\cH\otimes \cmplx^n)$ and that the family 
\[
	\{ \Pg^{(n)} (\E ; \cdot) \mid  \text{$n\in \natn$ and $\E$ is an ensemble on $\mathbf{D}(\cH\otimes \cmplx^n)$}\}
\] 
characterizes the poset $(\fCH , \ordCPpp)$.

\section{Main theorems and their proofs} \label{sec:main}
In this section we prove the following main theorems of this paper:

\begin{thm} \label{thm:main1}
Let $(V, V_+, \Omega)$ be a base-norm Banach space with $\dim V \geq 2$.
Then the following assertions hold.
\begin{enumerate}[1.]
\item \label{it:thm1.1}
Both $\dimord (\MfinO , \mathord\pp)$ and $\dimordR (\MfinO , \mathord\pp)$ are infinite.
\item \label{it:thm1.2}
If $V$ is separable in the norm topology, i.e.\ $V$ has a countable norm dense subset, then 
\begin{equation}
	\dimord (\MfinO , \mathord\pp) = \dimordR (\MfinO , \mathord\pp) = \aleph_0 ,
	\label{eq:main1}
\end{equation}
where $\aleph_0 = |\natn |$ denotes the cardinality of a countably infinite set.
\end{enumerate}
\end{thm}

\begin{thm} \label{thm:main2}
Let $\cH$ be a separable complex Hilbert space with $\dim \cH \geq 2 .$
Then 
\begin{gather}
	\dimord ( \Mfin (\DeH) ,\mathord\pp ) = \dimordR ( \Mfin (\DeH) ,\mathord\pp ) = \aleph_0 ,
	\label{eq:main2-1} \\
	\dimord ( \fC ( \cH) ,\mathord\CPpp ) = \dimordR ( \fC (\cH) ,\mathord\CPpp ) = \aleph_0 .
	\label{eq:main2-2}
\end{gather}
\end{thm}
\begin{rem} \label{rem:1}
The assumption $\dim V \geq 2$ in Theorem~\ref{thm:main1} holds if and only if the state space $\Omega$ contains at least two distinct points because $V$ is the linear span of $\Omega$ and $\Omega$ is contained in the hyperplane $\set{x \in V | \braket{u_\Omega , x} =1}$, which does not contain the origin.
Similarly the separability of $V$ in Theorem~\ref{thm:main1} is equivalent to that of $\Omega$.
\end{rem}

\begin{rem} \label{rem:2}
Theorem~\ref{thm:main1}.\ref{it:thm1.1} is proved by explicitly constructing a sequence of finite subsets of $\MfinO$ with arbitrarily large dimensions (see Lemmas~\ref{lemm:cbit} and \ref{lemm:main1}).
Indeed, it is proved in \cite{harzheim1970} that for every poset with an infinite order dimension we can always find such a sequence of finite subsets.
We give another simple proof of this fact in \ref{app:compact} based on Tychonoff\rq{}s theorem.
\end{rem}

\begin{rem} \label{rem:3} 
As shown in \cite{kuramochi2020compact}, we can define the set $\mathfrak{M} (\Omega)$ of post-processing equivalence classes of general (i.e.\ possibly continuous outcome) EVMs on $\Omega$.
Then we can easily show that Theorem~\ref{thm:main1}.\ref{it:thm1.1} is also valid for $\mathfrak{M} (\Omega)$ since $\MfinO$ is a subset of $\mathfrak{M} (\Omega)$, i.e.\ the finite-outcome EVM is a special kind of the general EVM.
We can also show that the proof of Theorem~\ref{thm:main1}.\ref{it:thm1.2} (Lemma~\ref{lemm:countable}) can be straightforwardly generalized to $\mathfrak{M} (\Omega)$ since the BSS theorem is also valid for general EVMs (\cite{kuramochi2020compact}, Theorem~1).
\end{rem}

\begin{rem} \label{rem:4}
A parametrized family $(P_\theta)_{\theta \in \Theta}$ of classical probabilities with a fixed sample space is called a statistical experiment, or a statistical model, and is one of the basic concepts in mathematical statistics~\cite{lecam1986asymptotic,torgersen1991comparison}.
As shown in \cite{kuramochi2020compact} (Appendix~D), we can identify the class of statistical experiments with a fixed parameter set $\Theta$ with the class of measurements (EVMs) on the input classical GPT $(\ell^1(\Theta), \ell^1_+(\Theta),\mathcal{P}(\Theta))$.
Based on this correspondence Theorem~\ref{thm:main1} straightforwardly applies to the post-processing order of statistical experiments.
\end{rem}

In the rest of this section, we prove Theorems~\ref{thm:main1}.\ref{it:thm1.1}, \ref{thm:main1}.\ref{it:thm1.2}, and \ref{thm:main2} in Sections~\ref{subsec:m1}, \ref{subsec:m1.2}, and \ref{subsec:m2}, respectively.

\subsection{Proof of Theorem~\ref{thm:main1}.\ref{it:thm1.1}} \label{subsec:m1}
The proof is split into some lemmas.
We first establish some general properties of the order dimensions necessary for the proof.

\begin{defi} \label{def:embedding}
Let $(S, \ordpreceq_1)$ and $(T, \ordpreceq_2)$ be posets.
A map $f \colon S \to T$ is called an \textit{order embedding} from $(S, \ordpreceq_1)$ into $(T, \ordpreceq_2)$ if 
$
	x \preceq_1 y 
$
if and only if $f (x) \preceq_2 f(y)$ for every $x,y \in S .$
An order embedding is necessarily an injection.
If such an order embedding exists, $(S, \ordpreceq_1)$ is said to be \textit{embeddable} into $(T, \ordpreceq_2) .$
If $(S, \ordpreceq_1)$ is embeddable into $(T, \ordpreceq_2) ,$ $(S, \ordpreceq_1)$ is order isomorphic to a subset of $T$ equipped with the restriction order of $\preceq_2 .$ 
\end{defi}

The following lemma is implicit in the literature~\cite{Dushnik1941,Hiraguti1955}, while here we give a proof for completeness.

\begin{lemm} \label{lemm:embedding}
Let $(S, \ordpreceq_1)$ and $(T, \ordpreceq_2)$ be posets.
Suppose that $(S, \ordpreceq_1)$ is embeddable into $(T, \ordpreceq_2) .$
Then $\dimord (S, \ordpreceq_1) \leq \dimord (T, \ordpreceq_2) $ and $\dimordR (S, \ordpreceq_1) \leq \dimordR (T, \ordpreceq_2) $ hold.
\end{lemm}
\begin{proof}
Let $g \colon S \to T$ be an order embedding and let $\mathcal{L}$ be a family of total orders on $T$ such that $\ordpreceq_2 = \bigcap \mathcal{L}$ and $| \mathcal{L} | = \dimord (T, \ordpreceq_2) .$
For each $L \in \mathcal{L} ,$ we define a binary relation $g^{-1}(L)$ on $S$ by
\[
	x g^{-1}(L) y : \defarrow g(x) L g(y) \quad (x, y \in S).
\]
Then, since $g$ is an injection, each $g^{-1}(L)$ $(L \in \mathcal{L})$ is a total order on $S .$
Moreover for every $x,y \in S$ we have
\begin{align*}
	x \preceq_1 y 
	& \iff g(x) \preceq_2 g(y)
	\\
	& \iff g(x) L g(y) \quad (\forall L \in \mathcal{L})
	\\
	& \iff  x g^{-1} (L) y \quad (\forall L \in \mathcal{L}) ,
\end{align*}
which implies that the family $\set{g^{-1} (L) | L \in \mathcal{L}}$ realizes $\preceq_1 .$ 
Then the first claim $\dimord (S, \ordpreceq_1) \leq \dimord (T, \ordpreceq_2) $ immediately follows from the definition of the order dimension.

Let $\F$ be a family of order monotones characterizing $(T , \ordpreceq_2)$
such that $| \F| =  \dimordR (T , \ordpreceq_2)$. 
Then for every $x,y \in S$ we have
\begin{align*}
	x \preceq_1 y 
	& \iff g(x) \preceq_2 g(y)
	\\
	& \iff f \circ g(x) \leq  f \circ g(y) \quad (\forall f \in \F).
\end{align*}
This implies that $\set{f\circ g}_{f \in \F}$ characterizes $(S , \ordpreceq_1 )$. 
From this the second claim $\dimordR (S, \ordpreceq_1) \leq \dimordR (T, \ordpreceq_2) $ immediately follows.
\end{proof}

Let $((S_i , \ordpreceq_i))_{i\in I}$ be an indexed family of posets.
We define a poset, called the \textit{direct product}, by
\begin{gather*}
	\bigotimes_{i\in I} (S_i , \ordpreceq_i) := \left(\prod_{i\in I} S_i , \ordpreceq \right) ,
	\\
	(x_i)_{i \in I} \preceq (y_i)_{i \in I} 
	: \defarrow
	[x_i \preceq_i y_i \quad (\forall i \in I)] .
\end{gather*}
For a poset $(S, \ordpreceq)$ we denote by $\dpc (S , \ordpreceq)$ the minimum cardinality $|I|$ of a family $((C_i, \ordpreceq_i))_{i\in I}$ of \textit{chains} such that $(S, \ordpreceq)$ is embeddable into the direct product $\bigotimes_{i\in I} (C_i , \ordpreceq_i) . $
Then it is known~\cite{milner1990note} that
\begin{equation}
	\dimord (S, \ordpreceq) = \dpc (S, \ordpreceq) 
	\label{eq:dpc}
\end{equation}
holds for every poset $(S, \ordpreceq)$.
From this we can show 
\begin{lemm} \label{lemm:leq}
Let $(S , \ordpreceq)$ be a poset.
Then $\dimord (S , \ordpreceq) \leq \dimordR (S , \ordpreceq) .$
\end{lemm}
\begin{proof}
Let $\mathcal{F}$ be a set of order monotones on $S$ such that $\mathcal{F}$ characterizes $\ordpreceq$ and $|\mathcal{F} | = \dimordR (S , \ordpreceq ) .$
Then the map
\begin{equation*}
	S \ni x \mapsto (f(x))_{f \in \F} \in \realn^{\F}
\end{equation*}
is an order embedding from $(S , \ordpreceq)$ into the direct product $\bigotimes_{f \in \F}(\realn , \mathord\leq ) ,$ where the order $\leq$ on the reals $\realn$ is the usual order.
Since $(\realn , \mathord\leq)$ is a chain, the claim follows from~\eqref{eq:dpc}.
\end{proof}

We now consider the specific base-norm Banach space $(\ell^2(\natn_2) , \ell^2_+(\natn_2) , \mathcal{P}(\natn_2))$ and the poset $(\Mfin (\mathcal{P}(\natn_2)), \ordpp) $.
We can and do identify $(\ell^2(\natn_2) , \ell^2_+(\natn_2) , \mathcal{P}(\natn_2))$ and $(\ell^\infty(\natn_2) , \ell^\infty_+ (\natn_2), u_{\mathcal{P}(\natn_2)} )$ with $(\realn^2 , \realn_+^2 , \Ocbit)$ and $(\realn^2 , \realn_+^2 , u_\cbit)$, respectively, where 
\begin{gather*}
	\realn_+ := [0,\infty) ,\\
	\Ocbit := \set{(p_0 , p_1) \in \realn^2 | p_0 , p_1 \geq 0 ,\, p_0 + p_1 =1 }, \\
	u_\cbit := (1,1).
\end{gather*}
Here the duality of $\realn^2$ and $\realn^{2 \ast} = \realn^2$ is given by
\begin{equation*}
	\braket{(a_0,a_1) , (b_0,b_1)}  := a_0 b_0 + a_1  b_1 
	\quad ((a_0,a_1) , (b_0 ,b_1) \in \realn^2) . 
\end{equation*}
The set $\Ocbit$ is a line segment and corresponds to the state space of a classical bit.

\begin{lemm} \label{lemm:cbit}
$(\Mfin (\Ocbit), \ordpp) $ is embeddable into $(\MfinO, \ordpp) .$
\end{lemm}
\begin{proof}
From the assumption $\dim V \geq 2$, we have $\dim V^\ast \geq 2$.
Thus there exists an element $a_0 \in V^\ast$ such that $(a_0, u_\Omega)$ is linearly independent.
We put 
\begin{gather*}
	 a := \left\| a_0 + \| a_0 \| u_\Omega \right\|^{-1}  (a_0 + \| a_0 \| u_\Omega),
	 \\
	 a^\prime := u_\Omega - a.
\end{gather*}
Then $a$ and $a^\prime$ are effects and $(a, a^\prime)$ is linearly independent.
We define a linear map $\Psi \colon \realn^2 \to V^\ast$ by
\begin{equation*}
	\Psi ((\alpha_0 , \alpha_1)) := \alpha_0 a + \alpha_1 a^\prime
	\quad ((\alpha_0,\alpha_1) \in \realn^2) .
\end{equation*}
Then, by the linear independence of $(a, a^\prime)$, $\Psi$ is injective.
Moreover, $\Psi$ is unital and positive, i.e.\ $\Psi (u_\cbit ) = u_\Omega$ and $\Psi (\realn_+^2) \subseteq V_+^\ast$ hold.

For each EVM $\oM = \oMseq $ on $\Ocbit$, we define 
\begin{equation*}
	\Psi (\oM) := (\Psi(\oM(j)))_{j=0}^{m-1}  \in V^{\ast m}.
\end{equation*}
Since $\Psi$ is unital and positive, $\Psi (\oM)$ is an EVM on $\Omega$.
Then for every EVMs $\oM = \oMseq$ and $\oN = \oNseq$ on $\Ocbit$ the equivalence
\begin{equation}
	\oM \pp \oN \iff \Psi (\oM) \pp \Psi (\oN) 
	\label{eq:cbitiff}
\end{equation}
holds.
The claim~\eqref{eq:cbitiff} can be shown as follows:
\begin{align*}
	&\oM \pp \oN \\
	&\iff \exists \text{$(p(j|k))_{j \in \Nm , \, k \in \Nn}$: Markov matrix s.t.\ }
	\left[\oM(j) = \sum_{k=0}^{n-1}p(j|k) \oN(k) \quad (\forall j \in \Nm) \right]
	\\
	&\iff \exists \text{$(p(j|k))_{j \in \Nm , \, k \in \Nn}$: Markov matrix s.t.\ }
	\left[\Psi(\oM(j)) = \sum_{k=0}^{n-1}p(j|k) \Psi(\oN(k))   \quad (\forall j \in \Nm) \right]
	\\
	&\iff \Psi (\oM) \pp \Psi (\oN) ,
\end{align*}
where the second equivalence follows from the injectivity of $\Psi$.
From \eqref{eq:cbitiff}, the map
\begin{equation*}
	\Mfin (\Ocbit) \ni [\oM] \mapsto [\Psi (\oM)] \in \MfinO 
\end{equation*}
is a well-defined order embedding, which completes the proof.
\end{proof}

From Lemmas~\ref{lemm:embedding} and \ref{lemm:cbit}, the proof of the infinite dimensionality of $(\MfinO , \ordpp)$ reduces to the case of the classical bit space $\Ocbit$.
This is done by proving that the following \textit{standard example of an $n$-dimensional poset} (\cite{Dushnik1941,Hiraguti1955}; \cite{BA18230795}, Chapter~1, \S~5) is embeddable into  $(\MfOc , \ordpp)$.

\begin{defi}[Standard example of an $n$-dimensional poset] \label{def:standard}
For each natural number $n \geq 2$, we define a poset $(S_n , \ordpreceq_n)$, called the standard example of an $n$-dimensional poset, as follows.
$S_n $ is a $2n$-element set given by $S_n := \{a_j\}_{j=0}^{n-1} \cup \{b_j \}_{j=0}^{n-1} $ and the order $\preceq_n$ is given by
\begin{equation*}
	\ordpreceq_n := \{(a_j , a_j) \}_{j=0}^{n-1} \cup \{(b_j , b_j) \}_{j=0}^{n-1} \cup
	\bigcup_{j=0}^{n-1} \bigcup_{k\in \Nn \setminus \{ j\}} \set{(a_j ,b_k)}
\end{equation*}
(see Figure~\ref{fig:hasse} for the Hasse diagram).
It is known~\cite{Dushnik1941,BA18230795} that $\dimord (S_n , \ordpreceq_n) = n$ holds for each $n \geq 2$.
\begin{figure}
\centering
\begin{tikzpicture}
	\node[circle,fill=white,draw=black,inner sep=0pt,minimum size=5pt,label=below:{$a_0$}] (a0) at (0,0) {};
	\node[circle,fill=white,draw=black,inner sep=0pt,minimum size=5pt,label=below:{$a_1$}] (a1) at (1,0) {};
	\node[circle,fill=white,draw=black,inner sep=0pt,minimum size=5pt,label=below:{$a_2$}] (a2) at (2,0) {};
	\node[circle,fill=white,draw=black,inner sep=0pt,minimum size=5pt,label=below:{$a_{n-1}$}] (an) at (4,0) {};
	\node[circle,fill=white,draw=black,inner sep=0pt,minimum size=5pt,label=above:{$b_0$}] (b0) at (0,2) {};
	\node[circle,fill=white,draw=black,inner sep=0pt,minimum size=5pt,label=above:{$b_1$}] (b1) at (1,2) {};
	\node[circle,fill=white,draw=black,inner sep=0pt,minimum size=5pt,label=above:{$b_2$}] (b2) at (2,2) {};
	\node[circle,fill=white,draw=black,inner sep=0pt,minimum size=5pt,label=above:{$b_{n-1}$}] (bn) at (4,2) {};
	\node (adots) at (3,0) {$\cdots$};
	\node (bdots) at (3,2) {$\cdots$};
	\draw (a0) -- (b1) -- (a2) -- (b0) -- (a1) -- (b2) -- (a0) -- (bn) -- (a1);
	\draw (a2) --(bn);
	\draw (an) -- (b0);
	\draw (an) -- (b1);
	\draw (an) -- (b2);
\end{tikzpicture}
\caption{The Hasse diagram of the standard example $(S_n , \ordpreceq_n)$ of an $n$-dimensional poset.}
\label{fig:hasse}
\end{figure}
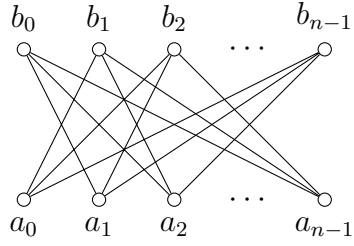
\end{defi}
In order to show that $(S_n , \ordpreceq_n)$ is embeddable into  $(\MfOc , \ordpp)$, we need the following notion of the direct mixture of EVMs and some related lemmas.

\begin{defi} \label{def:mix}
Let $\oM_j = (\oM_j(k))_{k=0}^{m_j -1}$ $(j = 0 ,1 , \dots, N-1)$ be EVMs on $\Omega$ and let $(p_j)_{j=0}^{N-1}$ be a probability distribution.
Then we define a $\sum_{j=0}^{N-1} m_j =: m$-outcome EVM $\oM = (\oM (k))_{k=0}^{m-1}$ by
\begin{equation*}
	\oM (k) := p_j \oM_j \left(k - \sum_{l=0}^{j-1} m_l \right)
	\quad \left(\text{when $\sum_{l=0}^{j-1} m_l \leq k < \sum_{l=0}^{j} m_l$}\right),
\end{equation*}
i.e.\
\begin{align*}
	\oM := &(p_0 \oM_0 (0), \dots, p_0 \oM_0 (m_0-1), p_1 \oM_1 (0) ,  \dots ,  p_1 \oM_1(m_1 -1) , \dots , \\ & p_{N-1} \oM_{N-1} (0) , \dots , p_{N-1} \oM(m_{N-1} -1)   ).
\end{align*}
This $\oM$ is called the \textit{direct mixture} and written as 
\begin{equation*}
\bigoplus_{j=0}^{N-1} p_j \oM_j
\end{equation*} 
or 
\begin{equation*}
	p_0 \oM_0 \oplus p_1 \oM_1 \oplus \dots \oplus p_{N-1} \oM_{N-1} .
\end{equation*}
In the operational language, the measurement corresponding to the direct mixture $\oM$ is realized as follows: we first generate a random number $j$ according to the probability distribution $(p_j)_{j=0}^{N-1}$, then perform $\oM_j$ which gives a measurement outcome $k \in \natn_{m_j},$ and finally record both $j$ and $k$.
The reader should not confuse the direct mixture with the ordinary mixture $( \sum_{j=0}^{N-1} p_j \oM_j (k) )_{k=0}^{m^\prime -1}$ (here we assume $m_j = m^\prime$ for all $j$), which is realized when we forget $j$ and record only $k$.
\end{defi}

The state discrimination probability $\Pg (\E ; \cdot)$ is affine with respect to the direct mixture as shown in the following lemma.

\begin{lemm}[cf.\ \cite{kuramochi2020compact}, Proposition~14.1] \label{lemm:affinity}
Let $\oM_j = (\oM_j(k))_{k=0}^{m_j -1}$ $(j = 0 ,1 , \dots, N-1)$ be EVMs on $\Omega$,
let $(p_j)_{j=0}^{N-1}$ be a probability distribution, and let $\E = (\rho_l)_{l=0}^{n-1}$ be an ensemble on $\Omega$.
Then 
\begin{equation}
	\Pg \left( \E ; \bigoplus_{j=0}^{N-1} p_j \oM_j  \right)
	= \sum_{j=0}^{N-1} p_j  \Pg (\E ; \oM_j) 
	\label{eq:affinity}
\end{equation}
holds.
\end{lemm}
\begin{proof}
By using \eqref{eq:mle}, we have
\begin{align*}
	\Pg \left( \E ; \bigoplus_{j=0}^{N-1} p_j \oM_j  \right)
	&= \sum_{j=0}^{N-1} \sum_{k=0}^{m_j-1} \max_{l \in \Nn} \braket{p_j \oM_j(k) , \rho_l}
	\\
	&=  \sum_{j=0}^{N-1} p_j \sum_{k=0}^{m_j-1} \max_{l \in \Nn} \braket{ \oM_j(k) , \rho_l}
	\\
	&=  \sum_{j=0}^{N-1} p_j \Pg (\E ; \oM_j),
\end{align*}
which proves \eqref{eq:affinity}.
\end{proof}

For each $ (s_0 , s_1) \in [0,1]^2(= \mathcal{E} (\Ocbit) ) $ we write as
\begin{equation*}
	\oA_{s_0 , s_1} := ((s_0,s_1) , u_\cbit - (s_0,s_1)) = ((s_0,s_1) , (1-s_0,1-s_1)) \in \evm_2 (\Ocbit),
\end{equation*}
which is the general form of a $2$-outcome EVM on $\Ocbit$.

\begin{lemm} \label{lemm:para}
Let $s_0 , s_1 , t_0 , t_1 \in [0,1]$.
Then $\oA_{s_0 , s_1} \pp \oA_{t_0 , t_1}$ if and only if there exist scalars $p, q \in [0,1] $ such that $(s_0 , s_1) = p (t_0 , t_1) + q (1-t_0 , 1-t_1)$, i.e.\ $(s_0 ,s_1)$ is inside the parallelogram $(0,0)$-$(t_0,t_1)$-$(1,1)$-$(1-t_0 ,1-t_1)$.
\end{lemm}
\begin{proof}
\begin{align*}
	&\oA_{s_0 , s_1} \pp \oA_{t_0 , t_1} \\
	&\iff \exists p,p^\prime , q, q^\prime \in [0, \infty) \quad \mathrm{s.t.} \quad 
	\begin{cases}
	p+p^\prime = q+ q^\prime =1, \\
	(s_0, s_1) = p(t_0 , t_1) + q (1-t_0 ,1-t_1), \\
	(1-s_0,1- s_1) = p^\prime (t_0 , t_1) + q^\prime (1-t_0 ,1-t_1)
	\end{cases}
	\\
	&\iff \exists p,q \in [0, 1] \quad \mathrm{s.t.} \quad 
	\begin{cases}
	(s_0, s_1) = p(t_0 , t_1) + q (1-t_0 ,1-t_1), \\
	(1-s_0,1- s_1) = (1-p) (t_0 , t_1) + (1-q) (1-t_0 ,1-t_1)
	\end{cases}
	\\
	&\iff \exists p,q \in [0, 1] \quad \mathrm{s.t.} \quad 
	(s_0, s_1) = p(t_0 , t_1) + q (1-t_0 ,1-t_1),
\end{align*}
where the final equivalence holds since the second condition 
\[
(1-s_0,1- s_1) = (1-p) (t_0 , t_1) + (1-q) (1-t_0 ,1-t_1)
\]
follows from the first condition 
\[
(s_0, s_1) = p(t_0 , t_1) + q (1-t_0 ,1-t_1). \qedhere
\]
\end{proof}

\begin{lemm} \label{lemm:parabola}
Let $0 < s < t < 1$. Then $\oA_{s, s^2}$ and $\oA_{t,t^2}$ are incomparable, i.e.\ neither $\oA_{s, s^2} \pp \oA_{t,t^2}$ nor $\oA_{t, t^2} \pp \oA_{s,s^2}$ holds.
\end{lemm}
\begin{proof}
By the strict convexity of the function $f(x) = x^2 ,$ the points $(s,s^2)$ and $(t,t^2)$ are respectively below the line segment $(0,0)$-$(t,t^2)$ and $(s,s^2)$-$(1,1)$. 
Therefore points $(s,s^2)$ and $(t,t^2)$ are respectively outside the parallelograms $(0,0)$-$(t, t^2)$-$(1,1)$-$(1-t ,1-t^2)$ and $(0,0)$-$(s, s^2)$-$(1,1)$-$(1-s ,1-s^2)$ (Figure~\ref{fig:para}).
Hence the claim follows from Lemma~\ref{lemm:para}.
\begin{figure}
\centering
\begin{tikzpicture}[domain=0:3.5,samples=200,>=stealth]
\node at (2 , -1) {(a)};
\draw[->] (0,0) -- (3.5,0) node[right] {$x$};
\draw[->] (0,0) -- (0,3.5) node[above] {$y$};
\draw plot (\x, {\x * \x / 3}) node[below right] {{\footnotesize$y=x^2$}};
\draw plot (\x , {3 - (3-\x) * (3-\x) /3} ) node[below right] {{\footnotesize$y= 1 - (1-x)^2$}};
\node[circle,fill=black,draw=black,inner sep=0pt,minimum size=3pt,label=above left:{\footnotesize$(1,1)$}] (u) at (3,3) {};
\node[circle,fill=black,draw=black,inner sep=0pt,minimum size=3pt,label=below left:{\footnotesize$(0,0)$}] (o) at (0,0) {};
\node[circle,fill=black,draw=black,inner sep=0pt,minimum size=3pt,label=below right:{\footnotesize$(s,s^2)$}] (s) at (1.3,1.69/3) {};
\node[circle,fill=black,draw=black,inner sep=0pt,minimum size=3pt,label=below right:{\footnotesize$(t,t^2)$}] (t) at (2,4/3) {};
\node[circle,fill=black,draw=black,inner sep=0pt,minimum size=3pt] (tp) at (3-2,3-4/3) {};
\node at (1.1, 3) {\footnotesize$(1-t,1-t^2)$};
\draw[->] (1,2.8) to [in=135,out=225]  (0.9,3-4/3);
\draw (o)--(t)--(u)--(tp)--(o);
\end{tikzpicture}
\begin{tikzpicture}[domain=0:3.5,samples=200,>=stealth]
\node at (2 , -1) {(b)};
\draw[->] (0,0) -- (3.5,0) node[right] {$x$};
\draw[->] (0,0) -- (0,3.5) node[above] {$y$};
\draw plot (\x, {\x * \x / 3}) node[below right] {{\footnotesize$y=x^2$}};
\draw plot (\x , {3 - (3-\x) * (3-\x) /3} ) node[below right] {{\footnotesize$y= 1 - (1-x)^2$}};
\node[circle,fill=black,draw=black,inner sep=0pt,minimum size=3pt,label=above left:{\footnotesize$(1,1)$}] (u) at (3,3) {};
\node[circle,fill=black,draw=black,inner sep=0pt,minimum size=3pt,label=below left:{\footnotesize$(0,0)$}] (o) at (0,0) {};
\node[circle,fill=black,draw=black,inner sep=0pt,minimum size=3pt,label=below right:{\footnotesize$(s,s^2)$}] (s) at (1.3,1.69/3) {};
\node[circle,fill=black,draw=black,inner sep=0pt,minimum size=3pt,label=below right:{\footnotesize$(t,t^2)$}] (t) at (2,4/3) {};
\node[circle,fill=black,draw=black,inner sep=0pt,minimum size=3pt] (sp) at (3-1.3, 3-1.69/3) {};
\node at (1.1, 3) {\footnotesize$(1-s,1-s^2)$};
\draw (o)--(s)--(u)--(sp)--(o);
\end{tikzpicture}
\caption{The points $(s,s^2)$ and $(t,t^2)$ are respectively outside the parallelograms $(0,0)$-$(t, t^2)$-$(1,1)$-$(1-t ,1-t^2)$ (a) and $(0,0)$-$(s, s^2)$-$(1,1)$-$(1-s ,1-s^2)$ (b).}
\label{fig:para}
\end{figure}
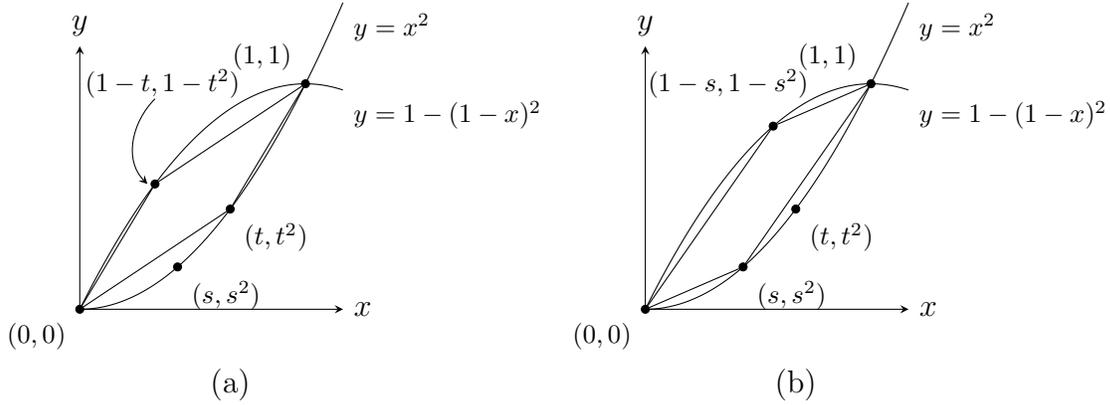
\end{proof}

We are now in a position to prove the following crucial lemma.

\begin{lemm} \label{lemm:main1}
$(S_n , \ordpreceq_n)$ is embeddable into $(\MfOc, \ordpp)$ for each natural number $n \geq 3$.
\end{lemm}
\begin{proof}
We write as $\oU := (u_\cbit) $, which is a trivial EVM on $\Ocbit$.
We define a map $f\colon S_n \to \MfOc$ by
\begin{gather*}
	s_j := 3^{j-n} , \\
	f(a_j) := [\oA^{(j)}], 
	\quad \oA^{(j)} := \frac{1}{n} \oA_{s_j , s_j^2} \oplus \frac{n-1}{n} \oU , \\
	f(b_j) :=[\oB^{(j)}] , \quad 
	\oB^{(j)} :=  \frac{1}{n} \oU \oplus \bigoplus_{k \in \Nn \setminus \{ j \}} \frac{1}{n} \oA_{s_k, s_k^2}
\end{gather*}
$(0 \leq j \leq n-1)$.
We establish the lemma by demonstrating that $f$ is an order embedding.
For this we have only to prove the following assertions:
for every $j, k \in \Nn$ with $j \neq k$
\begin{enumerate}
\item \label{it:lm1}
$\oA^{(j)}$ and $\oA^{(k)}$ are incomparable;
\item \label{it:lm2} 
$\oB^{(j)}$ and $\oB^{(k)}$ are incomparable;
\item \label{it:lm3}
$\oA^{(j)} \pp \oB^{(k)}$;
\item \label{it:lm4}
$\oB^{(k)} \pp \oA^{(j)}$ does not hold;
\item \label{it:lm5}
$\oA^{(j)}$ and $\oB^{(j)}$ are incomparable.
\end{enumerate}

\noindent
\textit{Proof of \eqref{it:lm1}.}
From Proposition~\ref{prop:bss} and Lemma~\ref{lemm:affinity} we have the following equivalences:
\begin{align*}
	&\oAj \pp \oAk  \\
	&\iff  \Pg (\E ; \oAj) \leq \Pg (\E ; \oAk) \quad (\text{$\forall \E$: ensemble})
	\\
	&\iff
	\frac{1}{n}\Pg (\E ; \oA_{s_j , s_j^2}) + \frac{n-1}{n} \Pg (\E; \oU)
	\leq 
	\frac{1}{n}\Pg (\E ; \oA_{s_k , s_k^2}) + \frac{n-1}{n} \Pg (\E; \oU) 
	\quad (\text{$\forall \E$: ensemble})
	\\
	&\iff 
	\Pg (\E ; \oA_{s_j , s_j^2}) \leq \Pg (\E ; \oA_{s_k , s_k^2}) \quad (\text{$\forall \E$: ensemble})
	\\ 
	&\iff  \oA_{s_j , s_j^2} \pp \oA_{s_k , s_k^2},
\end{align*}
but the last condition does not hold by Lemma~\ref{lemm:parabola}.
Hence $\oAj \pp \oAk$ does not hold.
Similarly $\oAk \pp \oAj$ does not hold.

\noindent
\textit{Proof of \eqref{it:lm2}.}
By using Proposition~\ref{prop:bss} and Lemma~\ref{lemm:affinity} we have
\begin{align*}
	&\oBj \pp \oBk  \\
	&\iff  \Pg (\E ; \oBj) \leq \Pg (\E ; \oBk) \quad (\text{$\forall \E$: ensemble})
	\\
	&\iff
	\frac{1}{n} \Pg (\E; \oU)+ \sum_{l \in \Nn \setminus\{ j\}}\frac{1}{n}\Pg (\E ; \oA_{s_l , s_l^2}) 
	\leq 
	\frac{1}{n} \Pg (\E; \oU)+ \sum_{l \in \Nn \setminus\{ k \}}\frac{1}{n}\Pg (\E ; \oA_{s_l , s_l^2}) 
	\quad (\text{$\forall \E$: ensemble})
	\\
	&\iff 
	\Pg (\E ; \oA_{s_k , s_k^2}) \leq \Pg (\E ; \oA_{s_j , s_j^2}) \quad (\text{$\forall \E$: ensemble})
	\\ 
	&\iff  \oA_{s_k , s_k^2} \pp \oA_{s_j , s_j^2},
\end{align*}
but the last condition does not hold by Lemma~\ref{lemm:parabola}.
Hence $\oBj \pp \oBk$ does not hold.
Similarly $\oBk \pp \oBj$ does not hold.

\noindent
\textit{Proof of \eqref{it:lm3}.}
By using Proposition~\ref{prop:bss} and Lemma~\ref{lemm:affinity} we have
\begin{align*}
	&\oAj \pp \oBk \\
	&\iff  \Pg (\E ; \oAj) \leq \Pg (\E ; \oBk) \quad (\text{$\forall \E$: ensemble})
	\\
	&\iff
	\frac{1}{n}\Pg (\E ; \oA_{s_j , s_j^2}) + \frac{n-1}{n} \Pg (\E; \oU)
	\leq \frac{1}{n} \Pg (\E; \oU)+ \sum_{l \in \Nn \setminus\{ k \}}\frac{1}{n}\Pg (\E ; \oA_{s_l , s_l^2}) 
	\quad (\text{$\forall \E$: ensemble})
	\\
	&\iff
	\sum_{l \in \Nn \setminus \{ j,k \}} (\Pg (\E ; \oA_{s_l , s_l^2}) - \Pg (\E; \oU))  \geq 0
	\quad (\text{$\forall \E$: ensemble}).
\end{align*}
The last condition holds because $\oU \pp \oM$ for every EVM $\oM$.
Thus $\oAj \pp \oBk$ holds.

\noindent
\textit{Proof of \eqref{it:lm4}.}
By using Proposition~\ref{prop:bss} and Lemma~\ref{lemm:affinity} we have
\begin{align*}
	&\oBk \pp \oAj \\
	&\iff  \Pg (\E ; \oBk ) \leq \Pg (\E ; \oAj) \quad (\text{$\forall \E$: ensemble})
	\\
	&\iff
	\frac{1}{n} \Pg (\E; \oU)+ \sum_{l \in \Nn \setminus\{ k \}}\frac{1}{n}\Pg (\E ; \oA_{s_l , s_l^2}) \leq \frac{1}{n}\Pg (\E ; \oA_{s_j , s_j^2}) + \frac{n-1}{n} \Pg (\E; \oU)
	\quad (\text{$\forall \E$: ensemble})
	\\
	&\iff
	\sum_{l \in \Nn \setminus \{ j,k \}} (\Pg (\E ; \oA_{s_l , s_l^2}) - \Pg (\E; \oU))  \leq 0
	\quad (\text{$\forall \E$: ensemble}).
\end{align*}
Since $\Pg (\E ; \oA_{s_l , s_l^2}) - \Pg (\E; \oU) \geq 0$, the last condition is equivalent to 
\begin{equation*}
	\Pg (\E ; \oA_{s_l , s_l^2}) = \Pg (\E ; \oU) \quad (\text{$\forall l \in \Nn \setminus \{ j,k \}$; $\forall \E$: ensemble}),
\end{equation*}
and therefore equivalent to $\oA_{s_l , s_l^2} \ppeq \oU $ $(\forall l \in \Nn \setminus \{ j, k \})$.
The EVM $\oA_{1,1}$ is trivial and hence post-processing equivalent to $\oU$.
Since $(s_l , s_l^2)$ is outside the line segment $(0,0)$-$(1,1)$, Lemma~\ref{lemm:para} implies that $\oA_{s_l , s_l^2} \pp \oA_{1,1}$, or equivalently $\oA_{s_l , s_l^2} \pp \oU$, does not hold.
Therefore $\oBk \pp \oAj$ does not hold.

\noindent
\textit{Proof of \eqref{it:lm5}.}
We first assume $\oAj \pp \oBj$ and derive a contradiction.
From the definition of the post-processing relation, the component $ (s_j/n , s_j^2/n)$ of the EVM $\oAj$ can be written as the following conic combination of the components of $\oBj$:
\begin{equation}
	\frac{1}{n}(s_j , s_j^2) 
	= \frac{1}{n}\sum_{l \in \Nn \setminus \{ j \} }  \left[ q_l (s_l  , s_l^2  )  + r_l (1-s_l , 1-s_l^2) \right]
	+\frac{r}{n} (1,1),
	\notag
\end{equation}
where $q_l, r_l , r \in [0,1]$.
This implies
\begin{equation}
	(s_j, s_j^2) 
	= \vec{u} + \vec{v},
	\label{eq:conic1} 
\end{equation}
where 
\begin{gather}
	\vec{u} := \sum_{l : \, 0 \leq l \leq j-1} q_l (s_l , 0), \notag
	\\
	\vec{v} := \sum_{l : \, 0 \leq l \leq j-1} q_l (0,s_l^2) + \sum_{l : \, j+1 \leq l \leq n-1} q_l (s_l , s_l^2)  + \sum_{l \in \Nn \setminus \{ j \}}  r_l (1-s_l , 1-s_l^2) + r (1,1) .
	\label{eq:vdef}
\end{gather}
Since we have
\begin{gather*}
	0 \leq \sum_{l : \, 0 \leq l \leq j-1} q_l s_l 
	\leq \sum_{l=0}^{j-1} s_l
	= \frac{3^{-n} (3^j -1)}{2}
	< \frac{3^{j - n} }{2} = \frac{s_j}{2}, 
\end{gather*}
$\vec{u}$ is in the line segment $L_j := \set{(t,0) | 0 \leq t \leq s_j /2} .$
On the other hand, $\vec{v}$ is in the convex cone 
\begin{align*}
	C_j &:= \set{ (\alpha s_{j+1} , \alpha s_{j+1}^2 + \beta) | \alpha , \beta \in [0,\infty)} \\
	&= \set{(x,y) \in \realn^2 | x \geq 0 , \, y \geq s_{j+1} x}
\end{align*}
generated by $(s_{j+1} , s_{j+1}^2)$ and $(0,1)$, where we put $s_n := 3^{n-n} = 1$.
This can be seen from that all the terms on the RHS of \eqref{eq:vdef} are in $C_j$
and $C_j$ is closed under conic combinations.
Therefore \eqref{eq:conic1} implies that 
\begin{equation}
	(s_j , s_j^2) \in L_j + C_j  = \set{ (x + t , y) | x \geq 0 , \, y \geq s_{j+1}x , \, 0 \leq t \leq s_j /2},
	\notag
\end{equation}
and hence there exists $ 0 \leq t \leq s_j/2$ such that
\begin{equation*}
	s_j^2 \geq s_{j+1} (s_j -t) .
\end{equation*}
This implies 
\begin{equation*}
	\frac{s_j}{2} \geq t \geq  (1 -s_j s_{j+1}^{-1})s_j = (1 - 3^{j-n - (j+1) + n}) s_j = \frac{2}{3}s_j >0,
\end{equation*}
which is a contradiction.
Thus $\oAj \pp \oBj$ does not hold.

We next assume $\oBj \pp \oAj$.
Then from \eqref{it:lm3} we have $\oAk \pp \oBj \pp \oAj$ and hence $\oAk \pp \oAj$.
This contradicts the incomparability of $\oAj$ and $\oAk$ which we have already proved in \eqref{it:lm1}.
Therefore $\oBj \pp \oAj$ does not hold.
\end{proof}

\noindent
\textit{Proof of Theorem~\ref{thm:main1}.\ref{it:thm1.1}.}
From Lemma~\ref{lemm:leq} we have
\begin{equation}
	\dimord (\MfinO , \ordpp) \leq \dimordR (\MfinO , \ordpp).
	\label{eq:mp1}
\end{equation}
From Lemma~\ref{lemm:embedding},
the order embedding results in Lemmas~\ref{lemm:cbit} and \ref{lemm:main1} imply
\begin{equation}
	n = \dimord (S_n , \ordpreceq_n) \leq \dimord (\MfOc, \ordpp) \leq \dimord (\MfinO, \ordpp)
	\label{eq:mp2}
\end{equation}
for each natural number $n \geq 3$.
Then the claim follows from \eqref{eq:mp1} and \eqref{eq:mp2}.
\qed

\subsection{Proof of Theorem~\ref{thm:main1}.\ref{it:thm1.2}} \label{subsec:m1.2}
Now we assume that $V$ is separable in the norm topology.
\begin{lemm} \label{lemm:countable}
There exists a countable family of order monotones that characterizes $(\MfinO , \ordpp) .$
\end{lemm}
\begin{proof}
For the proof we explicitly construct a countable family of order monotones that characterizes the post-processing order.
By the separability of $V$, the set of ensembles on $\Omega$ is also separable in the norm topology, which means that
we can take a sequence $\E^{(i)} = (\rho^{(i)}_j)_{j=0}^{N_i-1}$ $(i\in \natn)$ of ensembles on $\Omega$ such that for every ensemble $\E = (\rho_j)_{j=0}^{N-1}$ and every $\epsilon > 0$ there exists some $i \in \natn$ satisfying $N_i = N$ and
\begin{equation}
	\| \mathcal{E} - \E^{(i)} \| := \sum_{j=0}^{N-1} \| \rho_j - \rho_j^{(i)}  \| < \epsilon .
	\notag
\end{equation}
We now prove that the countable family $\{\Pg (\E^{(i)} ; \cdot)\}_{i \in \natn}$ of order monotones characterizes $(\MfinO , \ordpp) .$
For this, from Proposition~\ref{prop:bss}, we have only to prove that for every EVMs $\oM = \oMseq$ and $\oN = \oNseq$ on $\Omega$,
\begin{equation}
	\Pg (\E^{(i)} ; \oM) \leq \Pg (\E^{(i)} ; \oN ) 
	\quad (\forall i \in \natn)
	\label{eq:as1}
\end{equation}
implies
\begin{equation}
	\Pg (\E; \oM) \leq \Pg (\E ; \oN ) 
	\quad (\text{$\forall \E$: ensemble}) .
	\label{eq:con1}
\end{equation}
Assume \eqref{eq:as1}. 
We take an arbitrary ensemble $\E = (\rho_j)_{j=0}^{N-1}$ on $\Omega$ and $\epsilon >0 .$ 
By the density of $\{ \E^{(i)} \}_{i \in \natn}$ there exists some $i \in \natn $ such that $N_i = N$ and $\|\E - \E^{(i)} \| < \epsilon .$
Then from the definitions of $\Pg (\mathcal{E};\oM)$ and $\ordpp$ we have
\begin{align}
	\Pg (\E ; \oM ) 
	&=\sup_{\oA \in \evm (N; \oM )}
	\sum_{j=0}^{N-1} \braket{\oA(j) , \rho_j}
	\notag \\
	&=\sup_{\oA \in \evm (N; \oM )}
	\sum_{j=0}^{N-1} \left( \braket{\oA(j) , \rho_j^{(i)}} + \braket{\oA(j), \rho_j - \rho_j^{(i)}} \right)
	\notag \\
	& \leq \sup_{\oA \in \evm (N; \oM )}
	\sum_{j=0}^{N-1} \left(\braket{\oA(j), \rho_j^{(i)}} + \|\rho_j - \rho_j^{(i)} \| \right)
	\label{eq:der1}
	\\
	&\leq \sup_{\oA \in \evm (N; \oM )}
	\left( 
	\sum_{j=1}^N \braket{\oA(j), \rho_j^{(i)}} 
	+ \epsilon
	\right)
	\notag \\
	&=  \Pg (\E^{(i)}; \oM) + \epsilon  , \label{eq:ineq1}
\end{align}
where in deriving \eqref{eq:der1} we used the inequality
\[
	|\braket{f,x}| \leq \| f \| \|x\| 
	\quad (\text{$f \in V^\ast ,$ $x \in V$}) 
\]
and $\| \oA(j) \| \leq 1 .$
By replacing $\oM,$ $\E, $ and $\E^{(i)}$ in the above argument with $\oN,$ $\E^{(i)} ,$ and $\E,$ respectively, we also obtain
\begin{equation}
	\Pg (\E^{(i)} ; \oN) \leq \Pg (\E ; \oN) +  \epsilon.
	\label{eq:ineq2}
\end{equation}
From \eqref{eq:as1}, \eqref{eq:ineq1}, and \eqref{eq:ineq2}, we have
\begin{equation}
	\Pg (\E ; \oM) \leq \Pg(\E ; \oN) + 2 \epsilon .
	\label{eq:Pgineq}
\end{equation}
Since $\epsilon >0$ is arbitrary, \eqref{eq:Pgineq} implies $\Pg (\E ; \oM) \leq \Pg (\E ; \oN) ,$ which completes the proof. 
\end{proof}
\noindent 
\textit{Proof of Theorem~\ref{thm:main1}.\ref{it:thm1.2}.}
From \eqref{eq:mp1} and Theorem~\ref{thm:main1}.\ref{it:thm1.1}, we have
\begin{equation}
	\aleph_0 \leq \dimord (\MfinO, \ordpp) \leq \dimordR (\MfinO, \ordpp).
	\label{eq:mp3}
\end{equation}
On the other hand, Lemma~\ref{lemm:countable} implies
\begin{equation}
	\dimordR (\MfinO, \ordpp) \leq \aleph_0 .
	\label{eq:mp4}
\end{equation}
Then \eqref{eq:main1} follows from \eqref{eq:mp3} and \eqref{eq:mp4}. \qed

\subsection{Proof of Theorem~\ref{thm:main2}} \label{subsec:m2}
Now we assume that $\cH$ is a complex separable Hilbert space.
\begin{lemm} \label{lemm:embedding2}
$(\Mfin (\DeH{}) , \ordpp )$ is embeddable into $(\fCH , \ordCPpp)$.
\end{lemm}
\begin{proof}
For each EVM $\oM = \oMseq \in \evmfH $ we define a normal channel (called the quantum-classical channel) $\Gamma^\oM \colon \mathbf{B}(\cmplx^m) \to \BH$ by
\begin{equation*}
	\Gamma^\oM (a) := \sum_{j=0}^{m-1} \braket{\xi_j^{(m)} | a \xi_j^{(m)}} \oM(j) 
	\quad (a \in \mathbf{B}(\cmplx^m)) ,
\end{equation*}
where $(\xi_j^{(m)})_{j=0}^{m-1} $ is an orthonormal basis of $\cmplx^m .$
It is known that 
\begin{equation}
	\oM \pp \oN \iff \Gamma^\oM \CPpp \Gamma^\oN 
	\label{eq:iff3}
\end{equation}
holds for every finite-outcome EVMs $\oM $ and $\oN  $ on $\DeH $ (\cite{1751-8121-50-13-135302}, Proposition~1).
From \eqref{eq:iff3} it readily follows that the map
\begin{equation*}
	\Mfin (\DeH{}) \ni [\oM] \mapsto [\Gamma^\oM] \in \fCH
\end{equation*}
is a well-defined order embedding, which proves the claim.
\end{proof}

The proof of the following lemma is almost parallel to that of Lemma~\ref{lemm:countable}.

\begin{lemm} \label{lemm:countable2}
There exists a countable family of order monotones that characterizes $(\fCH , \ordCPpp) .$
\end{lemm}
\begin{proof}
Since $\cH \otimes \cmplx^n$ is separable for each $n \in \natn$, we can take a dense countable family $\E^{(n,i)} = (\rho_k^{(n, i)})_{k=0}^{N_{n,i}-1}$ $(i\in \natn)$ of ensembles on $\mathbf{D}(\cH \otimes \cmplx^n)$ in the sense that for every ensemble $\E = (\rho_k)_{k=0}^{m-1}$ on $\mathbf{D}(\cH \otimes \cmplx^n)$ and every $\epsilon >0$ there exists $i \in \natn$ such that $N_{n,i} = m$ and 
\begin{equation}
	\| \E - \E^{(n,i)} \|_1
	:=\sum_{k=0}^{m-1} \|  \rho_k  - \rho_k^{(n,i)} \|_1 < \epsilon .
	\notag
\end{equation}
We establish the lemma by demonstrating that the countable family $\{\Pg^{(n)} (\E^{(n,i)} ; \cdot )  \}_{n,i\in \natn}$ characterizes $(\fCH , \ordCPpp) .$
From Proposition~\ref{prop:qbss}, it suffices to show that 
\begin{equation}
	\Pg (\E^{(n,i)} ; \Gamma  \otimes \id_n) \leq \Pg (\E^{(n,i)} ; \Lambda \otimes \id_n )
	\quad (\forall i \in \natn)
	\label{eq:c1}
\end{equation}
implies
\begin{equation}
	\Pg (\E  ; \Gamma \otimes \id_n ) \leq \Pg (\E  ; \Lambda \otimes \id_n )
	\quad (\text{$\forall \E$: ensemble on $\mathbf{D}(\cH \otimes \cmplx^n)$})
	\label{eq:c2}
\end{equation}
for every $n \in \natn $ and every channels $\Gamma \colon \BK \to \BH$ and $\Lambda \colon \BJ \to \BH $.

Assume \eqref{eq:c1}.
Let $\E = (\rho_k)_{k=0}^{m-1}$ be an arbitrary ensemble on $\mathbf{D}(\cH \otimes \cmplx^n)$.
For any $\epsilon >0$ we can take $i\in \natn$ such that $N_{n,i} = m$ and $\| \E - \E^{(n,i)} \|_1 < \epsilon $.
Then we have
\begin{align}
	&\Pg (\E ; \Gamma \otimes \id_n) \notag \\
	&= \sup_{\oM \in \evm_m (\mathbf{D}(\cK \otimes \cmplx^n))} 
	\sum_{k=0}^{m-1} \braket{ (\Gamma \otimes \id_n) (\oM (k)) , \rho_k }
	\notag \\
	&=\sup_{\oM \in \evm_m (\mathbf{D}(\cK \otimes \cmplx^n))} 
	\sum_{k=0}^{m-1} \left( \braket{ (\Gamma \otimes \id_n) (\oM (k)) , \rho_k^{(n,i)}}
	+ \braket{ (\Gamma \otimes \id_n) (\oM(k)), \rho_k - \rho_k^{(n,i)}}
	\right)
	\notag \\
	&\leq 
	\sup_{\oM \in \evm_m (\mathbf{D}(\cK \otimes \cmplx^n))} 
	\sum_{k=0}^{m-1} \left( \braket{ (\Gamma \otimes \id_n) (\oM (k)), \rho_k^{(n,i)}}	+ \| \rho_k - \rho_k^{(n,i)} \|_1 \| (\Gamma \otimes \id_n) (\oM(k)) \|  
	\right)
	\notag \\
	&\leq 
	\sup_{\oM \in \evm_m (\mathbf{D}(\cK \otimes \cmplx^n))} 
	\left( \sum_{k=0}^{m-1}  \braket{ (\Gamma \otimes \id_n) (\oM (k)), \rho_k^{(n,i)}}	
	\right) + \epsilon  
	\label{eq:c2.5} \\
	&=
	\Pg (\E^{(n,i)} ; \Gamma \otimes \id_n) + \epsilon ,
	\label{eq:c3}
\end{align}
where we used $ \| (\Gamma \otimes \id_n) (\oM(k)) \| \leq 1$ in deriving \eqref{eq:c2.5}.
By replacing $\E $, $\E^{(n,i)}$, and $\Gamma$ in the above argument with $\E^{(n,i)}$, $\E$, and $\Lambda$, respectively, we also obtain 
\begin{equation}
	\Pg (\E^{(n,i)} ; \Lambda \otimes \id_n) \leq \Pg (\E ; \Lambda \otimes \id_n) + \epsilon .
	\label{eq:c4}
\end{equation}
From \eqref{eq:c1}, \eqref{eq:c3}, and \eqref{eq:c4}, we have
\begin{equation*}
	\Pg (\E ; \Gamma \otimes \id_n) \leq \Pg (\E ; \Lambda \otimes \id_n) + 2\epsilon .
\end{equation*}
Since $\epsilon >0$ is arbitrary, this implies $\Pg (\E ; \Gamma \otimes \id_n) \leq \Pg (\E ; \Lambda \otimes \id_n),$ which completes the proof of \eqref{eq:c1}$\implies$\eqref{eq:c2}.
\end{proof}
\noindent 
\textit{Proof of Theorem~\ref{thm:main2}.}
Since $\TsaH$ is separable in the trace norm topology and $\dim \TsaH \geq 4$, the first claim~\eqref{eq:main2-1} follows from Theorem~\ref{thm:main1}.\ref{it:thm1.2}.

Now we prove the second claim~\eqref{eq:main2-2}.
Lemma~\ref{lemm:leq} implies
\begin{equation}
	\dimord (\fCH , \ordCPpp) \leq \dimordR (\fCH , \ordCPpp) 
	\label{eq:prf2} .
\end{equation}
From Lemmas~\ref{lemm:embedding} and \ref{lemm:embedding2} and \eqref{eq:main2-1} we have
\begin{equation}
	\aleph_0 = \dimord (\Mfin (\DeH) , \ordpp) \leq \dimord (\fCH, \ordCPpp ).
	\label{eq:prf3} 
\end{equation}
From Lemma~\ref{lemm:countable2} we also have
\begin{equation}
	\dimordR  (\fCH , \ordCPpp) \leq \aleph_0 .
	\label{eq:prf2-2}
\end{equation}
Then \eqref{eq:main2-2} follows from \eqref{eq:prf2}, \eqref{eq:prf3}, and \eqref{eq:prf2-2}.
\qed

\section{Conclusion} \label{sec:conclusion}
In this paper we have evaluated the order and order monotone dimensions of the post-processing orders of measurements on an arbitrary non-trivial GPT $\Omega$ (Theorem~\ref{thm:main1}) and of quantum channels with a fixed input Hilbert space (Theorem~\ref{thm:main2}).
We found that all of these order dimensions are infinite.
Our results reveal that the post-processing order of measurements or quantum channels is qualitatively more complex than any order with a finite dimension, such as the adiabatic accessibility relation in thermodynamics or the LOCC convertibility relation of finite-dimensional bipartite pure states.

In the crucial step of the proof, we have explicitly constructed an order embedding from the standard example $(S_n , \ordpreceq_n)$ of an $n$-dimensional poset into the poset $(\MfOc , \ordpp)$ of the equivalence classes of finite-outcome EVMs on the classical bit space $\Ocbit$ for every $n \geq 3$ (Lemma~\ref{lemm:main1}).
We also note that the BSS-type theorems (Propositions~\ref{prop:bss} and \ref{prop:qbss}) played important role in the proofs of Lemmas~\ref{lemm:main1}, \ref{lemm:countable}, and \ref{lemm:countable2}.

As mentioned in the introduction, we can find many other important non-total orders in physics and quantum information, especially in quantum resource theories.
The present work is just the first step to evaluate the dimensions of these kinds of orders and it would be an interesting future work to investigate other orders in quantum information theory from the standpoint of the order dimension.

\ack
This work was supported by JSPS KAKENHI Grant No.~JP22K13977.

\appendix

\section{Proof of Proposition~\ref{prop:bss}} \label{app:bssevm}
In this appendix, we prove the BSS theorem for EVMs (Proposition~\ref{prop:bss}).

For the proof we use the following weak$\ast$ topology of $(V^\ast)^n$.
\begin{defi} \label{def:wstop}
Let $(V, \| \cdot \| )$ be a Banach space and let $n \in \natn$ be a natural number.
The product linear space $V^n$ equipped with the norm
\begin{equation*}
	\| (x_k)_{k=0}^{n-1} \| := \sum_{k=0}^{n-1} \|x_k \| 
	\quad ((x_k)_{k=0}^{n-1} \in V^n)
\end{equation*}
is a Banach space.
We can and do identify the continuous dual $(V^n)^\ast$ with the product linear space $(V^\ast)^n$ by the duality 
\begin{equation*}
	\braket{(f_k)_{k=0}^{n-1} , (x_k)_{k=0}^{n-1}} 
	:= \sum_{k=0}^{n-1} \braket{f_k , x_k}
	\quad ((x_k)_{k=0}^{n-1} \in V^n ,\, (f_k)_{k=0}^{n-1} \in (V^\ast)^n).
\end{equation*}
By this identification the uniform norm on $(V^\ast)^n$ is given by 
\begin{equation*}
	\| (f_k)_{k=0}^{n-1} \| := \max_{0 \leq k \leq n-1} \| f_k \| 
	\quad ((f_k)_{k=0}^{n-1} \in (V^\ast)^n).
\end{equation*}
The weak$\ast$ topology~\cite{dunfordschwartzvol1,schaefer1999topological} $\sigma ((V^\ast)^n, V^n)$ on $(V^\ast)^n$ is the weakest topology on $(V^\ast)^n$ such that the linear functional $(V^\ast)^n \ni (f_k)_{k=0}^{n-1} \mapsto \braket{(f_k)_{k=0}^{n-1} , (x_k)_{k=0}^{n-1}} \in \realn$ is continuous for all $(x_k)_{k=0}^{n-1} \in V^n$.
The weak$\ast$ topology is a locally convex Hausdorff topology~\cite{dunfordschwartzvol1,schaefer1999topological} on $(V^\ast)^n$.
For every weakly$\ast$ continuous linear functional $\xi \colon (V^\ast)^n \to \realn$ there exists a unique $(x_k)_{k=0}^{n-1} \in V^n$ such that 
\begin{equation*}
	\xi ((f_k)_{k=0}^{n-1}) = \braket{(f_k)_{k=0}^{n-1} , (x_k)_{k=0}^{n-1}}  
	\quad ((f_k)_{k=0}^{n-1} \in (V^\ast)^n).
\end{equation*}
\end{defi}

\noindent
\textit{Proof of Proposition~\ref{prop:bss}.}
We write as $\oM = \oMseq$ and $\oN = \oNseq$.
First assume $\oM \pp \oN$.
Then by the transitivity of $\pp$ we have $\evm (N; \oM) \subseteq \evm (N ; \oN)$ for every natural number $N$.
Therefore for every ensemble $\E = (\rho_l)_{l=0}^{N-1}$ we have
\begin{equation*}
	\Pg (\E ; \oM) 
	= \sup_{\oA \in \evm (N;\oM)} \sum_{l=0}^{N-1} \braket{\oA(l) , \rho_l}
	\leq \sup_{\oA \in \evm (N;\oN)} \sum_{l=0}^{N-1} \braket{\oA(l) , \rho_l}
	= \Pg (\E; \oN),
\end{equation*}
which proves the \lq\lq{}only if\rq\rq{} part of the claim.

To show the converse implication, we assume that $\oM \pp \oN$ does not hold and find an ensemble $\E$ satisfying $\Pg (\E ; \oM) > \Pg (\E;\oN)$.
From the definition of $\evm (m ; \oN)$ the assumption implies $\oM \not\in \evm (m ; \oN)$.
Since $\evm (m; \oN)$ is the image of the weakly$\ast$ continuous affine map
\begin{equation*}
	\mathbf{Markov}(m,n) \ni (p(j|k))_{j\in \Nm , k \in \Nn} \mapsto 
	\left( \sum_{k=0}^{n-1} p(j|k) \oN(k) \right)_{j=0}^{m-1} 
	\in (V^\ast)^m,
\end{equation*}
where $\mathbf{Markov}(m,n) \subseteq \realn^{m\times n}$ is the set of $(m\times n)$-Markov matrices, and $\mathbf{Markov}(m,n)$ is a compact convex set, 
$\evm (m;\oN)$ is a weakly$\ast$ compact convex subset of $(V^\ast)^m$.
Therefore by the Hahn-Banach separation theorem~\cite{dunfordschwartzvol1,schaefer1999topological} there exists an element $(v_j)_{j=0}^{m-1} \in V^m $ such that 
\begin{equation}
	\sum_{j=0}^{m-1} \braket{\oM(j) , v_j} 
	> \sup_{\oA \in \evm (m;\oN)} \sum_{j=0}^{m-1} \braket{\oA(j) , v_j}.
	\label{eq:sep1}
\end{equation}
Since the positive cone $V_+$ generates $V$, for each $0 \leq j \leq m-1$ we may write as $v_j = v_j^+ - v_j^-$ for some $v_j^\pm \in V_+$.
We define an ensemble $\E = (\rho_j)_{j=0}^{m-1}$ by 
\begin{gather*}
	v^- := \sum_{j=0}^{m-1} v_j^-, \\
	c:= \sum_{j=0}^{m-1}\braket{u_\Omega , v_j +v^-}
	= \sum_{j=0}^{m-1} \| v_j^+ \| + (m-1) \sum_{j=0}^{m-1}\| v_j^-\| ,
	\\
	\rho_j := c^{-1} \left(v_j + v^-   \right) =c^{-1} \left(v_j^+ + \sum_{k \in \Nm \setminus \{ j\} } v^-_k \right) \in V_+.
\end{gather*}
Here $c>0$ holds since \eqref{eq:sep1} implies that at least one $v_j^\pm$ is non-zero.
Note also that $m \geq 2$ holds since $m=1$ implies $\oM = (u_\Omega)  \pp \oN$.
Then we have 
\begin{align*}
	\Pg (\E;\oM) &\geq \sum_{j=0}^{m-1} \braket{\oM(j) , \rho_j} 
	\quad (\because \oM \in \evm (m;\oM))
	\\
	&= c^{-1} \sum_{j=0}^{m-1} \braket{\oM(j) , v_j} 
	+c^{-1} \sum_{j=0}^{m-1} \braket{\oM(j) , v^-}
	\\
	&= c^{-1} \sum_{j=0}^{m-1} \braket{\oM(j) , v_j} 
	+c^{-1} \braket{u_\Omega , v^-}
	\\
	&> c^{-1} \sup_{\oA \in \evm(m;\oN)} \sum_{j=0}^{m-1} \braket{\oA(j) , v_j} 
	+c^{-1} \braket{u_\Omega , v^-} \quad (\because \eqref{eq:sep1})
	\\
	&= \sup_{\oA \in \evm (m;\oN)} \sum_{j=0}^{m-1} \braket{\oA(j) , \rho_j} 
	\\
	&= \Pg (\E ; \oN),
\end{align*}
which completes the proof. \qed

\begin{rem} \label{rem:5}
As we can see from the construction of $\E$ in the \lq\lq{}if\rq\rq{} part of the above proof, for $m$-outcome EVM $\oM$, $\oM \pp \oN$ holds if and only if $\Pg (\E; \oM) \leq \Pg (\E ; \oN)$ for every $m$-element ensemble $\E$ \cite{Guff_2021}.
\end{rem}

\section{Proof of Proposition~\ref{prop:qbss}} \label{app:bss}
In this appendix we prove Proposition~\ref{prop:qbss}.
The proof is based on the following BSS-type theorem for channels.

\begin{lemm}[\cite{luczak2019}, Proposition~2] \label{lemm:luczak}
Let $\Gamma \colon \BK \to \BH$ and $\Lambda \colon \BJ \to \BH$ be normal unital positive maps.
Then the following conditions are equivalent.
\begin{enumerate}[(i)]
\item
$\Pg (\E ; \Gamma) \leq \Pg (\E ; \Lambda) $ holds for every ensemble $\E$ on $\DeH{}$.
\item
For every $m \in \natn$ and every EVM $\oM = \oMseq$ on $\DeK{} $, there exists an EVM $\oN = (\oN(j))_{j=0}^{m-1}$ on $\mathbf{D} (\cJ)$ such that
\begin{equation*}
	\Gamma (\oM(j)) = \Lambda (\oN(j)) \quad (j \in \Nm) .
\end{equation*}
\end{enumerate}
\end{lemm}
\noindent
\textit{Proof of Proposition~\ref{prop:qbss}.}
Assume \eqref{it:qbss1} and take a normal channel $\Psi \colon \BK \to \BJ$ such that $\Gamma = \Lambda \circ \Psi .$
Then for every $n \in \natn$ and every ensemble $\E = (\rho_j)_{j=0}^{m-1}$ on $\mathbf{D} (\cH \otimes \cmplx^n)$ we have
\begin{align*}
	\Pg (\E ; \Gamma \otimes \id_n )
	&= \sup_{\oM \in  \evm_m (\mathbf{D} (\cK \otimes \cmplx^n) )} \sum_{j=0}^{m-1} \tr [ \rho_j (\Gamma \otimes \id_n ) (\oM(j))  ] 
	\\
	&= \sup_{\oM \in  \evm_m (\mathbf{D} (\cK \otimes \cmplx^n) )} \sum_{j=0}^{m-1} \tr [ \rho_j (\Lambda \circ \Psi \otimes \id_n ) (\oM(j))  ]
	\\
	&=\sup_{\oM \in  \evm_m (\mathbf{D} (\cK \otimes \cmplx^n) )} \sum_{j=0}^{m-1} \tr [ \rho_j (\Lambda  \otimes \id_n ) \circ (\Psi \otimes \id_n) (\oM(j))  ]
	\\
	& \leq \Pg (\E ; \Lambda \otimes \id_n) ,
\end{align*}
where the last inequality follows because $((\Psi \otimes \id_n) (\oM(j)) )_{j=0}^{m-1}$ is an EVM on $\mathbf{D} (\cJ \otimes \cmplx^n )$.
Thus \eqref{it:qbss2} holds.

Conversely assume \eqref{it:qbss2}.
We first prove $\Gamma \CPpp \Lambda$ when $d := \dim \cK $ is finite.
We take orthonormal bases $(\xi_j)_{j\in \Nd}$ and $(e_j)_{j\in \Nd}$ of $\cK$ and  $\cmplx^d$, respectively.
We define the following maximally entangled vectors
\begin{equation}
	\eta_{k,m} := \frac{1}{\sqrt{d}}\sum_{j \in \Nd} e^{\frac{2\pi i}{d}jm}  \xi_j \otimes e_{j+k}
	\quad (k, m \in \Nd) ,
	\label{eq:bell}
\end{equation}
where the term $j + k$ means the sum modulo $d$.
Then $(\eta_{k,m})_{k,m \in \Nd}$ is an orthonormal basis of $\cK \otimes \cmplx^d$
and therefore $(\ket{\eta_{k,m}} \bra{\eta_{k,m}})_{k,m \in \Nd}$ is an EVM on $\mathbf{D} (\cK \otimes \cmplx^d )$.
From the assumption \eqref{it:qbss2} and Lemma~\ref{lemm:luczak}, there exists an EVM 
$(\widetilde{\oM}(k,m))_{k,m \in \Nd}$ on $\mathbf{D} ( \cJ \otimes \cmplx^d)$ such that
\begin{equation}
	(\Gamma \otimes \id_d ) (\ket{\eta_{k,m}} \bra{\eta_{k,m}}) 
	= (\Lambda \otimes \id_d) (\widetilde{\oM}(k,m))
	\quad 
	(k, m \in \Nd) .
	\label{eq:qbss1}
\end{equation}
From \eqref{eq:bell}, the LHS of \eqref{eq:qbss1} is evaluated to be
\begin{equation*}
	\frac{1}{d} \sum_{j,j^\prime \in \Nd}
	e^{\frac{2\pi i}{d} (j-j^\prime)m} \Gamma (\ket{\xi_j} \bra{\xi_{j^\prime}})
	\otimes \ket{e_{j+k}} \bra{e_{j^\prime + k}} .
\end{equation*}
On the other hand, if we write as 
\begin{equation*}
	\widetilde{\oM}(k,m) 
	=: \sum_{j, j^\prime \in \Nd} M^{(k,m)}_{j, j^\prime} \otimes \ket{e_j} \bra{e_{j^\prime}},
\end{equation*}
then the RHS of \eqref{eq:qbss1} is given by
\begin{equation*}
	\sum_{j, j^\prime \in \Nd} \Lambda (M^{(k,m)}_{j,j^\prime}) \otimes  \ket{e_j} \bra{e_{j^\prime}} .
\end{equation*}
Therefore \eqref{eq:qbss1} implies
\begin{equation}
	\frac{1}{d^2} \Gamma (\ket{\xi_{j}} \bra{\xi_{j^\prime}}) 
	=\frac{1}{d} 
	e^{\frac{2\pi i}{d}(j^\prime -j)m} \Lambda (M^{(k,m)}_{j+k, j^\prime +k})
	\quad
	(j,j^\prime ,k,m \in \Nd).
	\label{eq:qbss2}
\end{equation}
By taking the summation over $k,m \in \Nd$ in \eqref{eq:qbss2}, we obtain 
\begin{equation}
	\Gamma (\ket{\xi_{j}} \bra{\xi_{j^\prime}})  
	= \Lambda \circ \Phi (\ket{\xi_{j}} \bra{\xi_{j^\prime}}) ,
	\label{eq:qbss3}
\end{equation}
where $\Phi \colon \BK \to \BJ$ is a linear map defined by
\begin{equation}
	\Phi (\ket{\xi_{j}} \bra{\xi_{j^\prime}})
	:=\frac{1}{d} \sum_{k,m \in \Nd} e^{\frac{2\pi i}{d}(j^\prime -j)m} M^{(k,m)}_{j+k, j^\prime +k} .
	\label{eq:Phidef}
\end{equation}
We show that $\Phi$ is a unital CP map, from which $\Gamma  =\Lambda \circ \Phi \CPpp \Lambda$ follows.
From the completeness condition $\sum_{k,m\in \Nd} \widetilde{\oM}(k,m)  = \unit_{\cJ \otimes \cmplx^d} $ we have
\begin{equation}
	\sum_{k,m\in \Nd} M^{(k,m)}_{j,j^\prime} = \delta_{j,j^\prime} \unit_{\cJ}
	\quad (j,j^\prime \in \Nd),
	\label{eq:qbsssum}
\end{equation}
where $\delta$ denotes the Kronecker delta.
From \eqref{eq:Phidef} and \eqref{eq:qbsssum} we obtain
\begin{align*}
	\Phi (\unit_{\cK} )
	&= \sum_{j\in \Nd} \Phi (\ket{\xi_j} \bra{\xi_j})
	\\
	&=
	\frac{1}{d} 
	\sum_{j, k , m \in \Nd} M^{(k,m)}_{j+k, j +k}
	\\
	&=
	\frac{1}{d} 
	\sum_{j^\prime, k , m \in \Nd} M^{(k,m)}_{j^\prime, j^\prime}
	\\
	&=
	\frac{1}{d} 
	\sum_{j^\prime \in \Nd} \unit_{\cJ}
	\\
	&= \unit_{\cJ} .
\end{align*}
Thus $\Phi$ is unital.
From the positive semi-definiteness of $\widetilde{\oM}(k,m) ,$ the matrix
$( e^{\frac{2\pi i}{d}(j^\prime -j)m} M^{(k,m)}_{j+k, j^\prime +k})_{j,j^\prime \in \Nd}$ is positive semi-definite and hence the Choi matrix $(\Phi(\ket{\xi_j} \bra{\xi_{j^\prime}}))_{j,j^\prime \in \Nd}$ of $\Phi$ is also positive semi-definite by \eqref{eq:Phidef}.
Thus $\Phi$ is CP, which completes the proof of \eqref{it:qbss2}$\implies$\eqref{it:qbss1} when $ \cK$ is finite-dimensional.

Now we consider general $\cK .$
We take an orthonormal basis $(\zeta_i)_{i \in I}$ of $\cK$ and denote by $\FI$ the family of finite subsets of the index set $I$, which is directed by the set inclusion relation.
For each finite subset $F \in \FI $ we define projections
\begin{equation*}
	P_F := \sum_{i \in F} \ket{\zeta_i} \bra{\zeta_i} , \quad P_F^\perp := \unit_{\cK} - P_F
\end{equation*}
and a channel
\begin{equation*}
	\Gamma_F \colon \BKF \ni a \mapsto \Gamma (a + \tr (\rho_F a)  P_F^\perp ) \in \BH ,
\end{equation*}
where $\rho_F$ is a fixed density operator on $P_F \cK .$
Since 
\begin{equation*}
	\BKF \ni a \mapsto a + \tr (\rho_F a)  P_F^\perp \in \BK
\end{equation*}
is a channel, we have $\Gamma_F \CPpp \Gamma .$
Therefore, from the implication \eqref{it:qbss1}$\implies$\eqref{it:qbss2}, 
for every $n \in \natn $ and every ensemble $\E$ on $\mathbf{D} (\cH \otimes \cmplx^n)$ we obtain
\begin{equation}
	\Pg (\E ; \Gamma_F \otimes \id_n) \leq \Pg (\E ; \Gamma \otimes \id_n).
	\label{eq:qbssleq1}
\end{equation}
From the assumption \eqref{it:qbss2} we also have
\begin{equation}
	\Pg (\E ; \Gamma  \otimes \id_n) \leq \Pg (\E ; \Lambda \otimes \id_n) .
	\label{eq:qbssleq2}
\end{equation}
From \eqref{eq:qbssleq1} and \eqref{eq:qbssleq2} we obtain
\begin{equation*}
	\Pg (\E ; \Gamma_F \otimes \id_n) \leq
	 \Pg (\E ; \Lambda \otimes \id_n) .
\end{equation*}
Since $P_F \cK$ is finite-dimensional, from what we have shown in the last paragraph, the relation $\Gamma_F \CPpp \Lambda$ holds.
Thus for each $F\in \FI $ there exists a channel $\Xi_F \colon \BKF \to \BJ$ such that $\Gamma_F = \Lambda \circ \Xi_F$.
Then we have
\begin{equation}
	\Gamma_F  (P_F a P_F ) = \Lambda \circ \Xi_F (P_Fa P_F) 
	\label{eq:qbss4}
\end{equation}
for each $a\in \BH$.
Since $P_F a P_F \xrightarrow[F\in \FI]{\mathrm{uw}}  a$, $P_F^\perp  \xrightarrow[F\in \FI]{\mathrm{uw}} 0$, and $|\tr(\rho_FP_F a P_F )| \leq \| a \| $,
the normality of $\Gamma$ implies
\begin{align}
	\Gamma_F (P_F a P_F ) 
	&= \Gamma (P_F a P_F  )+ \tr(\rho_F P_F a P_F )\Gamma ( P_F^\perp) 
	\notag \\
	&\xrightarrow[F\in \FI]{\mathrm{uw}} \Gamma (a). 
	\label{eq:limit}
\end{align}
On the other hand, since $\| \Xi_F (P_F a P_F) \| \leq \| a \| , $
from the ultraweak compactness of a closed ball of $\BJ$ (Banach-Alaoglu theorem) and Tychonoff\rq{}s theorem, there exists a subnet $(\Xi_{F(\alpha)})_{\alpha \in A}$ such that the ultraweak limit 
\begin{equation}
	\uwlim_{\alpha \in A} \Xi_{F(\alpha)} (P_{F(\alpha)} a P_{F(\alpha)}   ) =: \Xi (a)
	\notag
\end{equation} 
exists for every $a\in \BK .$
Since the map
\begin{equation*}
	\BK \ni a \mapsto \Xi_{F } (P_{F } a P_{F }   ) \in \BJ
\end{equation*}
is a channel for each $F \in \FI$, so is $\Xi \colon \BK \to \BJ .$
Thus from \eqref{eq:qbss4} and \eqref{eq:limit}, for each $a\in \BK$ we obtain
\begin{align*}
	\Gamma (a)
	&= \uwlim_{\alpha \in A} \Gamma_{F(\alpha)} (P_{F(\alpha)}a P_{F(\alpha)})
	\\
	&=
	\uwlim_{\alpha \in A} \Lambda \circ \Xi_{F(\alpha)}  (P_{F(\alpha)}  a P_{F(\alpha)} )
	\\
	&= \Lambda \circ \Xi (a) ,
\end{align*}
where the last equality follows from the normality of $\Lambda .$
Therefore we have $\Gamma = \Lambda \circ \Xi \CPpp \Lambda , $ which completes the proof of \eqref{it:qbss2}$\implies$\eqref{it:qbss1}. \qed

\section{Compactness principle for order dimension} \label{app:compact}
In this appendix, we prove the following proposition.
\begin{prop}[\cite{harzheim1970}] \label{prop:compact}
Let $(S, \ordpreceq)$ be a poset and let $\mathbb{F}(S)$ denote the family of finite subsets of $S$.
Then the following assertions hold.
\begin{enumerate}[1.]
\item \label{it:compact1}
For every $n\in \natn ,$ the equivalence
\begin{equation}
	\dimord (S, \ordpreceq) \leq n
	\iff
	[\dimord (F, \ordpreceq \rvert_F) \leq n  
	\quad (\forall F \in \mathbb{F}(S)) ]
	\label{eq:compact1}
\end{equation}
holds.
\item \label{it:compact2}
If $ \dimord (S, \ordpreceq)$ is finite, then there exists a finite subset $F \in \mathbb{F}(S)$ such that  $\dimord (F, \ordpreceq \rvert_F) = \dimord (S, \ordpreceq)$.
\item \label{it:compact3}
$\dimord (S, \ordpreceq)$ is infinite if and only if for every $n \in \natn$ there exists a finite subset $F\in \mathbb{F}(S)$ such that $\dimord (F, \ordpreceq \rvert_F) > n .$
\end{enumerate}
\end{prop}

For the proof of this proposition, we use the following product topology on the power set.
\begin{defi} \label{def:topo}
Let $X$ be a set and let $\mathfrak{P}(X)$ denote the power set (i.e.\ the family of all the subsets of $X$).
As usual we identify $\mathfrak{P}(X)$ with the product set $\natn_2^X$ by the correspondence 
\begin{equation}
	\mathfrak{P}(X) \ni A \mapsto (1_A (x))_{x\in X} \in \natn_2^X,
	\label{eq:surj}
\end{equation}
where 
\begin{equation*}
	1_A (x) := 
	\begin{cases}
		1 & \text{if $x\in A$;} \\
		0 & \text{if $x\not\in A.$}
	\end{cases}
\end{equation*}
From the product topology on $\natn_2^X$ of the discrete topology on $\natn_2$, the bijection \eqref{eq:surj} induces a topology on $\mathfrak{P}(X) ,$ which we just call the product topology on $\mathfrak{P}(X) .$
By Tychonoff\rq{}s theorem, the product topology on $\mathfrak{P}(X)$ is a compact Hausdorff topology.
A net $(A_i)_{i\in I}$ on $\mathfrak{P}(X)$ converges to $A \in \mathfrak{P}(X)$ in the product topology if and only if $1_{A_i}(x) \to 1_A (x)$ for all $x\in X, $
or equivalently, if and only if 
\begin{equation}
	\limsup_{i \in I} A_i = \liminf_{i \in I} A_i = A,
	\notag
\end{equation}
where 
\begin{equation*}
	\limsup_{i \in I} A_i := \bigcap_{i \in I} \bigcup_{j \in I : j \geq i} A_j ,
	\quad
	\liminf_{i \in I} A_i := \bigcup_{i \in I} \bigcap_{j \in I : j \geq i} A_j .
\end{equation*} 
From this we can see that a net $(A_i)_{i\in I}$ on $\mathfrak{P}(X)$ is convergent if and only if $\limsup_{i \in I} A_i = \liminf_{i \in I} A_i .$
\end{defi}

\noindent
\textit{Proof of Proposition~\ref{prop:compact}.}
\begin{enumerate}[1.]
\item
The \lq\lq{}$\implies$\rq\rq{} part of the claim~\eqref{eq:compact1} is obvious from Lemma~\ref{lemm:embedding}.
To prove the converse implication, take arbitrary $n \in \natn$ and assume that $\dimord (F , \ordpreceq \rvert_F) \leq n$ for all $F \in \mathbb{F}(S) .$
Then for every $F \in \mathbb{F}(S) $ there exists a sequence $(L^{(k)}_F)_{k=0}^{n-1}$ of total orders on $F$ such that $\preceq \rvert_F$ is realized by $\{L^{(k)}_F\}_{k=0}^{n-1}$. (The sequence $( L^{(k)}_F)_{k=0}^{n-1}$ can contain identical elements if $\dimord (F , \ordpreceq \rvert_F) < n$.)
Since $L^{(k)}_F \subseteq F \times F \subseteq S \times S$, we may regard $(L^{(k)}_F)_{F \in \mathbb{F} (S)}$ as a net on $\mathfrak{P} (S\times S) .$
Then by the compactness of the product topology on $\mathfrak{P}(S \times S)$, there exist subnets $(L^{(k)}_{F(i)})_{i \in I}$ $(k \in \Nn)$ such that 
\[
	\limsup_{i \in I} L^{(k)}_{F(i)}
	= \liminf_{i \in I} L^{(k)}_{F(i)}
	=: L^{(k)} \in \mathfrak{P} (S\times S) 
	\quad
	(k\in \Nn) .
\]
We now show that $\{L^{(k)}\}_{k=0}^{n-1}$ is a family of total orders on $S$ that realizes $\preceq ,$ which proves the \lq\lq{}$\impliedby$\rq\rq{} part of the claim~\eqref{eq:compact1}.

We first prove that each $L^{(k)}$ is a total order on $S$.
Let $x,y,z \in S$ be arbitrary elements.
\begin{description}
\item[(Reflexivity).]
If $\{ x \} \subseteq F$, we have $(x,x) \in  L^{(k)}_F$ and therefore $(x,x) \in L^{(k)}_{F(i)}$ eventually holds. 
Thus $(x,x) \in L^{(k)} .$
\item[(Antisymmetry).]
Assume $(x,y) , (y,x) \in L^{(k)}$. 
Then $(x,y) , (y,x) \in L^{(k)}_{F(i)}$ for some $i \in I .$ 
By the antisymmetry of $L^{(k)}_{F(i)}$, this implies $x=y .$
\item[(Transitivity).]
By the transitivity of $L^{(k)}_F$ we have 
\begin{equation}
	1_{L^{(k)}_F}(x,y) 1_{L^{(k)}_F}(y,z) \leq 1_{L^{(k)}_F}(x,z) 
	\notag
\end{equation}
when $\{x,y,z \} \subseteq F$.
Since $\{x,y,z \} \subseteq F(i)$ eventually holds for $i \in I$, we have
\begin{align*}
	1_{L^{(k)}}(x,y) 1_{L^{(k)}}(y,z) 
	&= \lim_{i \in I} 1_{L^{(k)}_{F(i)}}(x,y) 1_{L^{(k)}_{F(i)}}(y,z) \\
	&\leq \lim_{i \in I} 1_{L^{(k)}_{F(i)}}(x,z)
	\\
	&= 1_{L^{(k)}} (x,z),
\end{align*}
which implies the transitivity of $L^{(k)}$.
\item[(Totality).]
From the totality of $L^{(k)}_F$ we have
\[
	1_{L^{(k)}_F} (x,y) + 1_{L^{(k)}_F} (y,x) \geq 1 
\]
when $\{ x,y\} \subseteq F .$
Therefore we have 
\begin{equation*}
	1_{L^{(k)} } (x,y) + 1_{L^{(k)} } (y,x)
	= \lim_{i \in I}[1_{L^{(k)}_{F(i)}} (x,y) + 1_{L^{(k)}_{F(i)}} (y,x)]
	\geq 1 ,
\end{equation*}
which implies the totality of $L^{(k)}$.
\end{description}
Thus we have shown that $L^{(k)}$ is a total order on $S$.

We now show that $\{L^{(k)}\}_{k=0}^{n-1}$ realizes $\ordpreceq$.
For this, we have only to prove 
\begin{equation}
	x \preceq y \iff [(x, y)  \in  L^{(k)} \quad (\forall k \in \Nn)]  
	\label{eq:compact3}
\end{equation}
for arbitrary $x,y \in S.$
Since $\{ L^{(k)}_F \}_{k=0}^{n-1}$ realizes $\preceq \rvert_F ,$ we have 
\begin{equation}
	1_{\ordpreceq}(x,y) = 1_{\ordpreceq \rvert_F}(x,y)
	= \prod_{k=0}^{n-1} 1_{L^{(k)}_F} (x,y)
	\notag 
\end{equation}
when $\{ x, y \} \subseteq F .$
Therefore 
\begin{align*}
	\prod_{k=0}^{n-1} 1_{L^{(k)}} (x,y) 
	= \lim_{i \in I} \prod_{k=0}^{n-1} 1_{L^{(k)}_{F(i)}} (x,y)
	= 1_{\ordpreceq} (x,y) ,
\end{align*}
which implies \eqref{eq:compact3}.
Thus  $\{L^{(k)}\}_{k=0}^{n-1}$ realizes $\ordpreceq$.

\item
From Lemma~\ref{lemm:embedding}, we have $\dimord (F , \ordpreceq \rvert_F) \leq \dimord (S, \ordpreceq) =: d \in \natn$ for every $F \in \mathbb{F}(S) .$
If there exists no $F\in \mathbb{F}(S) $ such that $\dimord (F , \ordpreceq \rvert_F) = d ,$ then $\dimord (F , \ordpreceq \rvert_F) \leq d -1$ for all $F \in \FI$ and hence \eqref{eq:compact1} implies $d = \dimord (S ,\ordpreceq) \leq d-1 , $ which is a contradiction.
Thus there exists $F\in \mathbb{F}(S) $ such that $\dimord (F , \ordpreceq \rvert_F) = d .$
\item
The claim is proved as 
\begin{align*}
	\text{$\dimord  (S ,\ordpreceq)$ is infinite}
	& \iff \forall n \in \natn , \, [ \dimord  (S ,\ordpreceq)  > n ] \\
	& \iff \forall n \in \natn , \, \exists F \in \mathbb{F} (S) \, [ \dimord  (F ,\ordpreceq \rvert_F)  > n ] ,
\end{align*}
where the last equivalence follows from \eqref{eq:compact1}.
\qed
\end{enumerate}

\section*{References}
\providecommand{\newblock}{}

\end{document}